\newif\ifcomments\commentstrue
\newif\ifsubmission\submissionfalse

\ifsubmission
  \documentclass{llncs}
\else
  \documentclass[11pt,a4paper]{article}
  \usepackage{fullpage}
  \usepackage{amsthm}

  \theoremstyle{plain}
  \newtheorem{theorem}{Theorem}[section]
  \newtheorem{proposition}[theorem]{Proposition}
  \newtheorem{lemma}[theorem]{Lemma}

  \theoremstyle{definition}
  \newtheorem{definition}[theorem]{Definition}

  \theoremstyle{definition}
  \newtheorem{remark}{Remark}
  
\fi

\usepackage{amssymb,amsmath,amsfonts}
\usepackage{enumerate}

\usepackage{graphicx}
\usepackage{enumitem}
\usepackage{algorithm}

\ifsubmission
  \setlist[itemize]{noitemsep}
  \setlist[itemize]{nolistsep}
  \setlist[enumerate]{noitemsep}
  \setlist[enumerate]{nolistsep}
\fi

\usepackage{color}
\definecolor{DarkBlue}{RGB}{0,0,150}
\usepackage[colorlinks=true,linkcolor=DarkBlue,citecolor=DarkBlue]{hyperref}
\usepackage{xcolor}
\usepackage{framed}
\usepackage{mathrsfs}
\usepackage[mathscr]{euscript}
\usepackage{lineno}

\usepackage{soul}
\newcommand{\subh}[1]{\par \vspace{4pt} \noindent \textbf{{#1}}}

\newtheorem{innercustomthm}{Theorem}
\newenvironment{customthm}[1]
  {\renewcommand\theinnercustomthm{#1}\innercustomthm}
  {\endinnercustomthm}

\newtheorem{innercustomlemma}{Lemma}
\newenvironment{customlemma}[1]
  {\renewcommand\theinnercustomlemma{#1}\innercustomlemma}
  {\endinnercustomlemma}

\newcommand{\NN}{\mathbb{N}}

\newcommand{\eps}{\varepsilon}

\newcommand{\compIndist}{\overset{c}{\approx}}
\newcommand{\defeq}{=}

\newcommand{\negl}{\mathsf{negl}}
\newcommand{\poly}{\mathsf{poly}}

\newcommand{\Gen}{\mathsf{Gen}}
\newcommand{\Enc}{\mathsf{Enc}}
\newcommand{\Dec}{\mathsf{Dec}}

\newcommand{\cC}{{\cal C}}

\newcommand{\cH}{{\cal H}}

\newcommand{\cM}{{\cal M}}

\newcommand{\sB}{{\sf B}}

\newcommand{\sD}{{\sf D}}

\newcommand{\Msg}{\cM}

\newcommand{\Ciph}{\cC}

\newcommand{\E}{\mathsf{E}}
\newcommand{\SKE}{\mathsf{SKE}}
\newcommand{\PKE}{\mathsf{PKE}}
\newcommand{\pk}{\mathsf{pk}}
\newcommand{\sk}{\mathsf{sk}}
\newcommand{\msg}{\mathsf{msg}}
\newcommand{\str}{\mathsf{str}}

\newcommand{\OurPhi}{\Phi^\star}
\newcommand{\OurXi}{\Xi^\star}

\DeclareMathAlphabet{\mathbfsf}{\encodingdefault}{\sfdefault}{bx}{n}

\newcommand{\Ext}{\mathsf{Ext}}
\newcommand{\TwoExt}{\mathsf{2Ext}}
\newcommand{\bo}{\boldsymbol 1}
\newcommand{\Gt}{\mathsf{GT}}
\newcommand{\Me}{H_{\infty}}
\newcommand{\sdist}[1]{\left\| #1 \right\|_{\sf s}}

\newcommand{\zo}{\{0,1\}}

\newcommand{\btau}{{\boldsymbol\tau}}

\newcommand{\mpair}{\boldsymbol{m}}

\newcommand{\state}{\mathfrak{s}}

\newcommand{\SetupPh}{\Phi}%{\Phi_{\rm Setup}}
\newcommand{\CommPh}{\Xi}%{\Phi_{\rm Comm}}

\newcommand{\round}[1]{{\color{red}#1}}
\newcommand{\roundcolorstring}{red}
\newcommand{\party}[1]{{\color{blue}#1}}
\newcommand{\partycolorstring}{blue}
\newcommand{\pr}[2]{{\party{#1},\round{#2}}}

\renewcommand{\emptyset}{\varnothing}

\newcommand{\ie}{{\it i.e.}}

\begin{document}
\pagestyle{plain}

\ifsubmission
\author{}
\institute{}
\else
  \title{{\bf How to Subvert Backdoored Encryption:} \\ {\Large Security Against
  Adversaries that Decrypt All Ciphertexts}}
  \author{Thibaut Horel\thanks{Supported, in part, by the National Science Foundation under grants CAREER IIS-1149662, and CNS-1237235, by the Office of Naval Research under grants YIP N00014-14-1-0485 and N00014-17-1-2131, and by a Google Research Award.}\\Harvard University
  \and Sunoo Park\thanks{Supported by the Center for Science of
  Information STC (CSoI),
  an NSF Science and Technology Center (grant agreement CCF-0939370),
  MACS project NSF grant CNS-1413920, and
  a Simons Investigator Award Agreement dated 2012-06-05.}\\MIT
  \and Silas Richelson\\UC Riverside
  \and Vinod Vaikuntanathan\thanks{Supported in part by NSF Grants CNS-1350619, CNS-1414119 and CNS-1718161, Alfred P. Sloan Research Fellowship, Microsoft Faculty Fellowship and a Steven and Renee Finn Career Development Chair from MIT.}\\MIT}
\fi
\date{}
\maketitle

\begin{abstract}
In this work, we examine the feasibility of secure and undetectable point-to-point communication in a world where governments can read all the encrypted communications of their citizens.  We consider a world where the only permitted method of communication is via a government-mandated encryption scheme, instantiated with government-mandated keys.  Parties cannot simply encrypt ciphertexts of some other encryption scheme, because citizens caught trying to communicate outside the government's knowledge (\emph{e.g.}, by encrypting strings which do not appear to be natural language plaintexts) will be arrested.  The one guarantee we suppose is that the government mandates an encryption scheme which \emph{is} semantically secure against outsiders: a perhaps reasonable supposition when a government might consider it advantageous to secure its people's communication against foreign entities. But then, what good is semantic security against an adversary that holds all the keys and has the power to decrypt?

We show that even in the pessimistic scenario described, citizens \emph{can} communicate securely and undetectably.  In our terminology, this translates to a positive statement: all semantically secure encryption schemes support \emph{subliminal communication}.  Informally, this means that there is a two-party protocol between Alice and Bob where the parties exchange ciphertexts of what appears to be a normal conversation even to someone who knows the secret keys and thus can read the corresponding plaintexts.  And yet, at the end of the protocol, Alice will have transmitted her secret message to Bob.  Our security definition requires that the adversary not be able to tell whether Alice and Bob are just having a normal conversation using the mandated encryption scheme, or they are using the mandated encryption scheme for subliminal communication.

Our topics may be thought to fall broadly within the realm of \emph{steganography}: the science of hiding secret communication within innocent-looking messages, or \emph{cover objects}.  However, we deal with the non-standard setting of an adversarially chosen distribution of cover objects (\emph{i.e.}, a stronger-than-usual adversary), and we take advantage of the fact that our cover objects are ciphertexts of a semantically secure encryption scheme to bypass impossibility results which we show for broader classes of steganographic schemes. We give several constructions of subliminal communication schemes under the assumption that key exchange protocols with pseudorandom messages exist (such as Diffie-Hellman, which in fact has truly random messages). Each construction leverages the assumed semantic security of the adversarially chosen encryption scheme, in order to achieve subliminal communication.
%An alternate perspective on our results would be an affirmative answer to the question: is there any meaningful guarantee provided by semantic security against an adversary with the power to decrypt?
\end{abstract}

%%%%%%% Vinod's Block %%

\thispagestyle{empty}
\newpage
\pagenumbering{roman}
\newpage
\pagenumbering{arabic}

%%%%%%%%%%%%%%%%%%%%%%%%

\section{Introduction}

Suppose that we lived in a world where the government wished to read all the
communications of its citizens, and thus decreed that citizens must not
communicate in any way other than by using a specific, government-mandated
encryption scheme with government-mandated keys. Even face-to-face
communication is not allowed: in this Orwellian world, anyone who is caught
speaking to another person will be arrested for treason. Similarly, anyone
whose communications appear to be hiding information will be arrested:
\emph{e.g.}, if the plaintexts encrypted using the government-mandated scheme
are themselves ciphertexts of a different encryption scheme. However, the one
assumption that we entertain in this paper, is that the government-mandated
encryption scheme is, in fact, semantically secure: this is a tenable
supposition with respect to a government that considers secure encryption to be
in its interest, in order to prevent foreign powers from spying on its
citizens' communications.

A natural question then arises: is there any way that the citizens would be able to communicate in a fashion undetectable to the government, based only on the semantic security of the government-mandated encryption scheme, and \emph{despite the fact that the government knows the keys and has the ability to decrypt all ciphertexts}?\footnote{We note that one could, alternatively, consider an adversary with decryption capabilities arising from possession of some sort of ``backdoor.'' For the purposes of this paper, we opted for the simpler and still sufficiently expressive model where the adversary's decryption power comes from knowledge of all the decryption keys.} What can semantic security possibly guarantee in a setting where the adversary has the private keys?

This question may appear to fall broadly within the realm of \emph{steganography}: the science of hiding secret communications within other innocent-looking communications (called ``cover objects''), in an undetectable way. Indeed, it can be shown that if two parties have a shared secret, then based on slight variants of existing techniques for \emph{secret-key steganography}, they can conduct communications hidden from the government.\footnote{We refer the reader to Section~\ref{sec:related} for more details.}

However, the question of whether two parties who have never met before can
conduct hidden communications is more interesting. This is related to the
questions of \emph{public-key steganography} and \emph{steganographic key
exchange} which were both first formalized by von Ahn and Hopper~\cite{AH04}.
Public-key steganography is inadequate in our setting since exchanging or
publishing public keys is potentially conspicuous and thus is not an option in
our setting. All prior constructions of steganographic key exchange require the
initial sampling of a public random string that serves as a public parameter of
the steganographic scheme. Intuitively, in these constructions, the public
random string can be thought to serve the purpose of selecting a specific
steganographic scheme from a family of schemes \emph{after} the adversary has
chosen a strategy. That is, the schemes crucially assume that the adversary
(the dystopian government, in our story above) cannot choose its covertext
distribution as a function of the public parameter.

It is conservative and realistic to expect a malicious adversary to choose the
covertext distribution {\em after} the honest parties have decided on their
communication protocol (including the public parameters). After all, malice
never sleeps~\cite{Micali16}. Alas, we show that if the covertext distribution
is allowed to depend on the communication protocol, steganographic
communication is impossible. In other words, for every purported steganographic
communication protocol, there is a covertext distribution (even one with high
min-entropy) relative to which the communication protocol fails to embed
subliminal messages. The relatively simple counterexample we construct is
inspired by the impossibility of deterministic extraction.

\subh{Semantic Security to the Rescue?}
However, this impossibility result does not directly apply to our setting, as
our covertext distribution is restricted to be a sequence of ciphertexts (that
may encrypt arbitrary messages).  Moreover, the ciphertexts are semantically
secure against entities that are not privy to the private keys. We define the
notion of a \emph{subliminal communication scheme}
(Definition~\ref{def:subliminal}) as a steganographic communication scheme
where security holds relative to covertext distributions that are guaranteed to
be ciphertexts of some semantically secure encryption scheme. Is there a way to
use semantic security to enable subliminal communication?

Our first answer to this question is negative. In particular, consider the following natural construction: first, design an extractor function $f$; then, to subliminally transmit a message bit $b$, sample encryptions $c$ of a (even adversarially prescribed) plaintext $m$ using independent randomness every time, until $f(c) = b$. There are two reasons this idea does not work. First, if the plaintext bit $b$ is not random, the adversary can detect this fact by simply applying the extractor function $f$ to the transmitted covertext. Second, the government can pick an adversarial (semantically secure) encryption scheme where the extractor function $f$ is constant on all ciphertexts; this is again similar to the impossibility of deterministic extraction.

Nevertheless, we show how to circumvent these difficulties and use the semantic
security of the underlying (adversarial) encryption scheme and construct
a subliminal communication scheme. 

\begin{theorem}[Informal version of Theorem
	\ref{thm:subl}]\label{thm:construction1_informal} Under the decisional
	Diffie-Hellman (DDH) assumption---or any other assumption that gives rise
	to a key exchange protocol with messages indistinguishable from
	random---there is a subliminal communication scheme which allows the
	transmission of $O(\log\kappa)$ many bits per ciphertext after a setup
	phase of $\widetilde O(\log\kappa)$ ciphertexts ($\kappa$ is the security
	parameter).
%that requires no trusted setup and
%allows the transmission of arbitrarily polynomially many hidden bits
%for any cover distribution consisting of ciphertexts of a semantically secure encryption scheme.
\end{theorem}

We then show how to improve our first construction to reduce the length of the
setup phase under additional assumptions.

%\sunoo{in the following paragraph we might want to emphasize that the problem actually reduces to deterministic extraction from ciphertexts,
%rather than explaining solely in terms of reducing to overt communication schemes as it's currently presented.} \vinod{I think this is handled by the current writing, but please check.}

\subsection{Overview of Our Construction}
The first idea in our construction is implicit in essentially all the works in
steganography starting from \cite{Sim83}: namely, to achieve subliminal
communication of arbitrary messages, it is sufficient to be able to
undetectably communicate \emph{uniformly randomly distributed strings of one's choice}. In other words, Alice samples a string $r$ which is randomly distributed,
produces some ciphertext(s) to be sent to Bob, such that Bob is able to learn
$r$ from them, and yet a PPT eavesdropper Eve who
sees the entire communication transcript cannot distinguish between the
following two cases:
\begin{enumerate}
\item Alice is indeed sending (hereafter, ``embedding'') random strings to Bob, or
\item Alice is producing ciphertexts using the unmodified government-mandated
	encryption algorithm, without embedding such random strings.
\end{enumerate}
To be more precise, the indistinguishability requirement holds for any given
(adversarially specified) distribution $\mathcal{M}$ of message sequences that
Alice may choose to encrypt using the government-mandated encryption scheme.
Notice that this does not preclude that Eve may be able to learn $r$ and
indeed, our constructions do allow an eavesdropper to learn the embedded
strings. 
Given the ability to undetectably communicated randomly distributed strings,
Alice and Bob can then embed to each other the messages of a key-exchange
protocol with randomly distributed messages (such as Diffie-Hellman) to establish a shared secret, and then embed to each other ciphertexts of a secret-key encryption scheme with pseudorandom ciphertexts, using the established secret as the key.

All known constructions of such \emph{undetectable random string embedding} rely on the sampling of a public random seed after the adversarial strategy is fixed. In this paper, however, we are interested in bootstrapping hidden communications from the very ground up, and we are not willing to assume that the parties start from a state where such a seed is already present.

We observe that the ability to embed randomly distributed strings \emph{of one's choice} --- rather than, \emph{e.g.}, to apply a deterministic function to ciphertexts of the government-mandated encryption scheme, and thereby obtain randomly distributed strings which the creator of the ciphertexts did not choose --- is crucial to the above-outlined scheme.
The notion of undetectably embedding \emph{exogenous} random strings ---
\emph{i.e.}, strings that are randomly distributed outside of Alice's control, but both Alice and Bob can read them --- is seemingly much weaker, and certainly cannot be used to embed key exchange messages or secret-key ciphertexts.
However, we observe that this weaker primitive turns out to be achievable, for
our specific setting, without the troublesome starting assumption of a public random seed.
We identify a method for embedding \emph{exogenous} random strings into ciphertexts of an adversarially chosen encryption scheme (interestingly, our method does not generalize to embedding into arbitrary min-entropy distributions). We then exploit this method to allow the communicating parties to establish a random seed --- from which point they can proceed to embed random strings \emph{of their choice}, as described above.

In building this weaker primitive,
in order to bypass our earlier-described impossibility result, we
extract from two ciphertexts at a time, instead of one. We begin with the following simple idea: for each consecutive pair of
ciphertexts $c$ and $c'$, a single hidden (random) bit $b$ is defined by $b
= f(c,c')$ where $f$ is some two-source extractor. It is initially unclear why this
should work because (1) $c$ and $c'$ are encryptions of messages $m$ and $m'$
which are potentially dependent, and two-source extractors are not guaranteed to
work without independence; and (2) even if this difficulty could  be
overcome, ciphertexts of semantically secure encryption scheme can have
min-entropy as small as $\omega(\log \kappa)$ (where $\kappa$ is the security
parameter) and no two-source extractor known to this day can extract from such a small min-entropy.

We overcome difficulty (1) by relying on the semantic security of the ciphertexts of the adversarially chosen encryption scheme. Paradoxically, even though the adversary knows the decryption key, we exploit the fact that semantic security still holds against the \emph{extractor}, which does not have the decryption key.
The inputs in our case are ciphertexts which are not necessarily
independent, but semantic security implies that they are computationally indistinguishable from being
independent. Thus, the output of $f(c,c')$ is
pseudorandom. Indeed, when $f$ outputs a single bit (as in our construction), the output is also
statistically close to random. The crucial point here is that the semantic
security of the encryption scheme is used not against the government, but
rather against the extraction function $f$.

Our next observation, to address difficulty (2), is that the ciphertexts are not only computationally
independent, but they are also computationally indistinguishable from i.i.d. In
particular, each pair of ciphertexts is indistinguishable from a pair of encryptions of $0$, by semantic security. Based on this observation, we can use
a very simple ``extractor'', namely, the greater-than function ${\sf GT}$.  In fact,
${\sf GT}$ is an extractor with two input sources, whose output bit has
negligible bias when the sources have $\omega(\log\kappa)$
min-entropy and are \emph{independently and identically distributed} (this
appears to be a folklore observation; see, \emph{e.g.}, \cite{BIW04}).
Because of the last condition, $\Gt$ is not a true two-source extractor according to standard definitions, but is still suitable for our setting.

By repeatedly extracting random bits from pairs of consecutive ciphertexts
using $\Gt$, Alice and Bob can construct a shared random string $s$. Note that
in this process, Alice and Bob generate ciphertexts using the unmodified
government-mandated encryption scheme, so the indistinguishability requirement
clearly holds. We stress again that $s$ is also known to a passive eavesdropper
of the communication. This part of our construction, up to the construction of
the string $s$, is presented in details in Section~\ref{sec:establish-seed}.
From there, constructing a subliminal communication scheme is not hard: Alice
and Bob use $s$ as the seed of a strong seeded extractor to subliminally
communicate random strings \emph{of their choice} as explained in
Section~\ref{sec:overt}. The complete description of our protocol is given in
Section~\ref{sec:full-protocol}.

\subsection{Improved Constructions for Specific Cases}

While our first construction has the advantage of simplicity, the initial phase
to agree on shared random string (using the ${\sf GT}$ function) transmits only
one hidden bit per ciphertext of the government-mandated encryption scheme.
A natural question is whether this rate of transmission can be improved.  We
show that if the government-mandated encryption scheme is \emph{succinct} in
the sense that the ciphertext expansion factor is at most $2$, then it is
possible to improve the rate of transmission in this phase to $O(\log\kappa)$
hidden bits per ciphertext using an alternative construction based on the
extractor from \cite{improv}. In other words, our first result showed that if
the government-mandated encryption scheme is semantically secure, we can use it
to communicate subliminally; the second result shows that if the
government-mandated encryption scheme is efficient, that is even better for us,
in the sense that it can be used for more efficient subliminal communication.

\begin{theorem}[Informal version of Theorem \ref{thm:cc}]\label{thm:construction2_informal}
If there is a secure key exchange protocol whose message distribution is
	pseudorandom, then there is a subliminal communication scheme in which
	a shared seed is established in two exchanges of ciphertexts of a succinct
	encryption scheme.
\end{theorem}

Theorem~\ref{thm:construction1_informal} exploited the specific nature of the
cover object distribution in our setting (specifically, that a sequence of
encryptions of arbitrary messages is indistinguishable from an i.i.d. sequence
of encryptions of zero). Theorem~\ref{thm:construction2_informal} exploits an
additional consequence of the semantic security of the government-mandated
encryption scheme: if it is succinct, then ciphertexts are computationally
indistinguishable from sources of high min-entropy (\emph{i.e.}, they have large HILL-entropy).

It may be possible to use more advanced two-source extractors to work with a larger class of government-mandated encryption schemes (with larger expansion factors); however, the best known such extractors have an inverse polynomial error rate~\cite{CZ16} (whereas our construction's extractor has negligible error). Consequently, designing a subliminal communication protocol using these extractors seems to require additional ideas, and we leave this as an open problem.

Finally, we show yet another approach in cases where the distribution of
``innocent'' messages to be encrypted under the government-mandated encryption
scheme has a certain amount of conditional min-entropy.  For such cases, we
construct an alternative scheme that leverages the semantic security of the
encryption scheme in a rather different way: namely, the key fact for this
alternative construction is that (in the absence of a decryption key)
a ciphertext appears independent of the message it encrypts. In this case,
running a two-source extractor on the message and the ciphertext works. The
resulting improvement in the efficiency of the scheme is comparable to that of
Theorem~\ref{thm:construction2_informal}.

\begin{theorem}[Informal version of Theorem \ref{thm:mc}]\label{thm:construction3_informal}
If there is a secure key exchange protocol whose message distribution is
	pseudorandom, then there is a subliminal communication scheme: 
\begin{itemize}
	\item for any cover distribution consisting of ciphertexts of a semantically secure encryption scheme, if the \emph{innocent message distribution} $\mathcal{M}$ has conditional min-entropy rate $1/2$, or
	\item for any cover distribution consisting of ciphertexts of a semantically secure \emph{and succinct} encryption scheme,
	if the \emph{innocent message distribution} $\mathcal{M}$ has conditional min-entropy $\omega(\log\kappa)$.
\end{itemize}
In both cases, the shared seed is established during the setup phase in only
	two exchanges of ciphertexts.
\end{theorem}

We conclude this introductory section with some discussion of our results in a wider context.

\subh{On Our Modeling Assumptions.}
Our model considers a relatively powerful adversary that, for example, has the ability
to choose the encryption scheme using which all parties must communicate, and to decrypt all such communications.
We believe that this can be very realistic in certain scenarios, but 
it is also important to note the limitations that our model places on the adversary.

The most obvious limitation is that the encryption scheme chosen by the adversary must be semantically secure
(against third parties that do not have the ability to decrypt).
Another assumption is that citizens are able to run algorithms of their choice on their own computers
without, for instance, having every computational step monitored by the government.
Moreover, citizens may use encryption randomness of their choice
when producing ciphertexts of the government-mandated encryption scheme:
in fact, this is a key fact that our construction exploits.
Interestingly, secrecy of the encryption randomness from the adversary is irrelevant:
after all, the adversary can always choose an encryption scheme where the encryption randomness
is recoverable given the decryption key.
Despite this, the ability of the encryptor to choose the randomness to input to the encryption algorithm
can be exploited---as by our construction---to allow for subliminal communication.
%\vinod{But you don't allow **common** public randomness.}

\subh{The Meaning of Semantic Security when the Adversary Can Decrypt.}
In an alternate light, our work may be viewed as asking the question:
\emph{what guarantee, if any, does semantic security provide against adversary in possession of the decryption key?}
Our results find, perhaps surprisingly, that some meaningful guarantee is still provided by semantic security
even against an adversary is able to decrypt: more specifically, that \emph{any} communication channel
allowing transmission of ciphertexts can be leveraged to allow for undetectable communications
between two parties that have never met.
From this perspective, our work may be viewed as the latest in a scattered series of recent works
that consider what guarantees can be provided by cryptographic primitives that are somehow ``compromised''---examples of recent works in this general flavor are cited in Section \ref{sec:related} below.

\subh{Concrete Security Parameters.}
From a more practical perspective, it may be relevant to consider that the government in our hypothetical Orwellian scenario
would be incentivized to opt for an encryption scheme with the least possible security level
so as to ensure security against foreign powers. In cases where the government considers itself
to have more computational power than foreign adversaries (perhaps by a constant factor),
this could create an interesting situation where the security parameter
with which the government-mandated scheme must be instantiated
is \emph{below} what is necessary to ensure security against the government's own computational power.

Such a situation could be risky for citizens' hidden communications:
intuitively, our constructions guarantee indistinguishability \emph{against the citizens' own government}
between an ``innocent'' encrypted conversation and one which is carrying hidden subliminal messages.
However, the distinguishing advantage in this indistinguishability game
\emph{depends on the security parameter} of the government-mandated encryption scheme.
Thus, it could be that the two distributions are far enough apart for the citizens' own government to distinguish
(though not for foreign governments to distinguish).
We observe that citizens cognizant of this situation can further reduce the distinguishing advantage beyond that provided by our basic construction, using the standard technique of amplifying the proximity of a distribution (which is far from random) to uniformly random, by taking the XOR of several samples from the far-from-random distribution.

Having outlined this potential concern and solution,
in the rest of the paper we will disregard these issues in the interest of clarity of exposition,
and present a purely asymptotic analysis.

\subh{Open Problems.}
%\vinod{essentially, removing the assumptions -- the existence of trusted randomness etc. -- to the extent possible, or prove negative results.}
Our work suggests a number of open problems. A natural one is the extent to which the modeling assumptions that this work makes --- such as the ability of honest encryptors to use true randomness for encryption --- can be relaxed or removed,
while preserving the ability to communicate subliminally. For example, one could imagine yet another alternate universe, in which the hypothetical Orwellian government
not only mandates that citizens use the prescribed encryption scheme, but also that their encryption randomness must be derived from a specific government-mandated pseudorandom generator.

The other open problems raised by our work are of a more technical nature and
better understood in the context of the specific details of our constructions;
for this reason we defer their discussion to Section~\ref{sec:open}.

\subsection{Other Related Work}
\label{sec:steg}
\label{sec:related}

The scientific study of steganography was initiated by Simmons more than thirty years ago \cite{Sim83},
and is the earliest mention of the term ``subliminal channel''
referring to the conveyance of information in a
cryptosystem's output in a way that is different from the intended output,\footnote{This phrasing is loosely borrowed from \cite{YY97}.} of which we are aware.
Subsequent works such as \cite{Cac98,Mit99,ZFKPPWWW98}
initially explored information-theoretic treatments of steganography,
and then Hopper, Langford, and von Ahn \cite{HLA02}
gave the first complexity-theoretic (secret-key) treatment almost two decades later.
Public-key variants of steganographic notions---namely,
public-key steganography and steganographic key exchange---were
first defined by \cite{AH04}.
There is very little subsequent literature on public-key steganographic primitives;
one notable example is by Backes and Cachin \cite{BC05}, which considers public-key steganography against
active attacks (their attack model, which is stronger than that of \cite{AH04},
was also considered in \cite{HLA02} but had never been applied to the public-key setting).

The alternative perspective of our work as addressing the question of whether any sort of secret communication can be achieved
via transmission of ciphertexts of an adversarially designed cryptosystem alone fits into a scattered series of recent works
that consider what guarantees can or cannot be provided by compromised cryptographic primitives.
For example, Goldreich \cite{Gol02}, and later, Cohen and Klein \cite{CK16},
consider what unpredictability guarantee is achieved by the classic GGM construction \cite{GGM86} when the traditionally secret seed is known;
Austrin et al. \cite{ACMPS14} study whether certain cryptographic primitives can be secure even in the presence of an adversary that has limited ability to tamper with honest parties' randomness;
Dodis et al. \cite{DGGJR15} consider what cryptographic primitives can be built based on backdoored pseudorandom generators;
and Bellare, Jaeger, and Kane \cite{BJK15} present attacks that work against any symmetric-key encryption scheme,
that completely compromise security by
undetectably corrupting the algorithms of the encryption scheme (such attacks might, for example, be feasible if an adversary could generate
a bad version of a widely used cryptographic library and install it on his target's computer).

The last work mentioned above, \cite{BJK15}, is actually part of the broader field of kleptography,
originally introduced by Young and Yung \cite{YY97,YY96b,YY96a}, which is also
relevant context for the present work.
Broadly speaking, a \emph{kleptographic attack} ``uses cryptography against
cryptography'' \cite{YY97} --- \emph{i.e.},
changes the behavior of a cryptographic system in a fashion undetectable to an honest user with black-box access to the cryptosystem, such that the use of the modified system leaks some secret information (\emph{e.g.}, plaintexts or key material)
to the attacker who performed the modification.
An example of such an attack might be to modify
the key generation algorithm of an encryption scheme such that an adversary in possession of a
``back door'' can derive the private key from the public key, yet an honest user finds the
generated key pairs to be indistinguishable from correctly produced ones.
Kleptography has enjoyed renewed research activity since \cite{BPR14} introduced a formal model of a specific type of kleptographic attack called \emph{algorithm substitution attacks} (ASAs), motivated by recent revelations suggesting that intelligence agencies have successfully implemented attacks of this nature at scale. Recently, \cite{BL17} formalized an equivalence between certain variants of ASA and steganography.

Our setting differs significantly from kleptography in that the encryption algorithms are public and not tampered with (\emph{i.e.}, adhere to a purported specification),
and in fact may be \emph{known} to be designed by an adversarial party.

\section{Preliminaries}

\paragraph{Notation.}
$\kappa$ is the security parameter throughout. PPT means ``probabilistic polynomial time.''
$[n]$ denotes the set $\{1,\ldots, n\}$. 
$U_n$ is a uniform variable over
$\zo^n$, independent of every other variable in this paper. We write $X\sim Y$
to express that $X$ and $Y$ are identically distributed. Given two variables
$X$ and $Y$ over $\zo^k$, we denote by $\sdist{X - Y}$ the statistical distance defined by:
	\begin{displaymath}
			\sdist{X-Y}
			=\frac{1}{2}\sum_{x\in\zo^k}\big|\Pr[X=x] - \Pr[Y=x]\big|
			\ifsubmission
			\;.
			\else
			=\max_{S\subseteq\zo^k}\big|\Pr[X\in S]-\Pr[Y\in S]\big|.
			\fi
	\end{displaymath}
For a random variable $X$, we define the min-entropy of $X$ by $\Me(X)=-\log\max_{x} \Pr[X=x]$. The collision probability is ${\sf CP}(X)\defeq\sum_x\Pr[X=x]^2$.

%We will make repeated use of the standard inequalities:
%\begin{displaymath}
%	\frac{1}{2^{2\Me(X)}}\leq CP(X)\leq \frac{1}{2^{\Me(X)}}
%\end{displaymath}
%\begin{displaymath}
%	CP(X)-\frac{1}{2^n}
%	\leq\sdist{X-U_n}^2
%	\leq 2^n \left(CP(X)-\frac{1}{2^n}\right)
%\end{displaymath}

\subsection{Encryption and Key Exchange}\label{sec:prelims-crypto}
We assume familiarity with the standard notions of semantically secure
public-key and private-key encryption, and key exchange. This subsection defines notation and additional terminology.

\subh{Public-Key Encryption.}
We use the notation
$\E=(\E.\Gen,\E.\Enc,\E.\Dec)$ 
for the public-key encryption scheme mandated by the adversary.

\subh{Secret-key Encryption.}
We write $\SKE=(\SKE.\Gen,\SKE.\Enc,\SKE.\Dec)$ to denote
a secret-key encryption scheme.
We define a \emph{pseudorandom secret-key encryption scheme}
to be a secret-key encryption scheme whose ciphertexts are indistinguishable from random.
%More precisely, the pseudorandomness guarantee is that for any polynomial $p$
%and any sequence of messages $m_1,\dots,m_p$, for %$\sk^*\gets\SKE.\Gen(1^{\sec})$,
%the following distributions are computationally indistinguishable:
%$$\big(\SKE.\Enc(\sk^*,m_1),\dots,\SKE.\Enc(\sk^*,m_p)\big)\mbox{ and }(U_\xi)^p\ ,$$
%where $\xi$ is the ciphertext length.
It is a standard result that 
pseudorandom secret-key encryption schemes 
can be built from one-way functions.

\subh{Key Exchange.}
We define a \emph{pseudorandom key-exchange protocol}
to be a key-exchange protocol whose transcripts are distributed
indistinguishably from random messages.
Recall that the standard security guarantee for key-exchange protocols
requires that $(T,K)\compIndist(T,K_\$)$,
where $T$ is a key-exchange protocol transcript, $K$ is the shared key established in $T$,
and $K_\$$ is a random unrelated key. A pseudorandom key-exchange protocol
instead requires that $(T,K)\compIndist(U,K_\$)$ 
where $U$ is the uniform distribution over strings of the appropriate length.

Most known key agreement protocols are pseudorandom; 
in fact, most have truly random messages.
This is the case, for example, 
for the classical protocol of Diffie and Hellman \cite{DH76}.

\subsection{Extractors} 

We will need the following definitions of two-source and seeded extractors.

\begin{definition}
	\label{def:tse}
	The family $\TwoExt: \zo^n \times \zo^m \to\zo^\ell$ is
	a \emph{$(k_1, k_2, \eps)$ two-source extractor}
	if for all $\kappa\in\NN$ and for all pairs $(X,Y)$ of independent random
	variables over $\zo^{n(\kappa)}\times\zo^{m(\kappa)}$ such that $\Me(X)\geq
	k_1(\kappa)$ and $\Me(Y)\geq k_1(\kappa)$, it holds that:
	\begin{equation}
		\label{eq:tse}
		\sdist{\TwoExt_{\kappa}(X,Y)-U_{\ell(\kappa)}}\leq \eps(\kappa).
	\end{equation}
	We say that $\TwoExt$ is \emph{strong} w.r.t. the first input if it
	satisfies the following stronger property:
	\begin{displaymath}
		\sdist{(X,\TwoExt_{\kappa}(X,Y))-(X,U_{\ell(\kappa)})}\leq \eps(\kappa).
	\end{displaymath}
	A strong two-source extractor w.r.t. the second input is defined analogously.
	Finally, we say that $\TwoExt$ is a $(k, \eps)$ \emph{same-source}
	extractor if $n=m$ and \eqref{eq:tse} is only required to hold when $(X,
	Y)$ is a pair of i.i.d. random variables with $\Me(X)=\Me(Y)\geq
	k(\kappa)$.
\end{definition}

\begin{definition}
	\label{def:seeded-ext}
	The family $\Ext: \zo^n \times \zo^m \to\zo^\ell$ is
	a \emph{$(k, \eps)$ seeded extractor}
	if for all $\kappa\in\NN$ and any random
	variable $X$ over $\zo^{m(\kappa)}$ such that $\Me(X)\geq
	k(\kappa)$, it holds that:
	\begin{equation*}
		\sdist{\Ext_{\kappa}(U_{n(\kappa)},X)-U_{\ell(\kappa)}}\leq \eps(\kappa).
	\end{equation*}
	We say moreover that $\Ext$ is \emph{strong} if it
	satisfies the following stronger property:
	\begin{displaymath}
		\sdist{(U_{n(\kappa)},\Ext_{\kappa}(U_{n(\kappa)},X))-(U_{n(\kappa)},U_{\ell(\kappa)})}\leq \eps(\kappa).
	\end{displaymath}
\end{definition}

\section{Subliminal Communication}
\label{sec:subl}

\subh{Conversation Model.} The protocols we will construct take place over
a communication between two parties $P_0$ and $P_1$ alternatingly sending each other
ciphertexts of a public-key encryption scheme. \emph{W.l.o.g.}, we
assume that $P_0$ initiates the communication, and that communication occurs
over a sequence of \emph{exchange-rounds} each of which comprises two sequential messages:
in each exchange-round, one party $P_b$ sends a message to $P_{1-b}$
and then $P_{1-b}$ sends a message to $P_b$. Let
$m_\pr{b}{i}$ denote the plaintext message sent by $P_b$ to $P_{1-b}$ in exchange-round $i$,
and let $\mpair_i=(m_\pr{0}{i},m_\pr{1}{i})$ denote the pair of messages exchanged.
For $i\geq 1$, let us denote by
$$\btau_\pr{0}{i}=(\mpair_1,\dots,\mpair_{i-1})
\mbox{ and }
\btau_\pr{1}{i}=(\mpair_1,\dots,\mpair_{i-1},m_\pr{0}{i})$$ 
the plaintext transcripts available to $P_0$ and $P_1$ respectively during 
exchange-round $i$, in the case when $P_0$ sends the first message in exchange-round $i$.\footnote{If instead $P_1$ spoke first in round $i$, then $\btau_\pr{0}{i}$ would contain $m_\pr{1}{i}$, and $\btau_\pr{1}{i}$ would not contain $m_\pr{0}{i}$.} We define $\btau_\pr{0}{0}$ and $\btau_\pr{1}{0}$ to be empty 
lists (\emph{i.e.}, empty starting transcripts).
(Note that when a notation contains both types of subscripts,
we write the subscripts denoting the party and round in 
\party{\partycolorstring} and \round{\roundcolorstring} respectively, to improve readability.)

Recall that our adversary has the power to decrypt all ciphertexts under 
its chosen public-key encryption scheme $\E$. Intuitively, it is therefore
important that the plaintext conversation between $P_0$ and $P_1$ appears innocuous (and does not,
for example, consist of ciphertexts of another encryption scheme).  
To model this, we assume the existence of a next-message distribution $\mathcal{M}$, 
which outputs a next innocuous message given the transcript of the plaintext conversation so
far. This is denoted by $m_\pr{b}{i}\gets\cM(\btau_\pr{b}{i})$.

%\paragraph{Notation.} We use $::$ to denote the operation of appending to a list, 
%\emph{e.g.}, $\btau_i=m_i::\btau_{i-1}$.

\begin{remark}
We emphasize that our main results make no assumptions at all on the distribution $\mathcal{M}$,
and require only that the parties have oracle access to their own next-message distributions.
Our main results hold in the presence of \emph{arbitrary} message distributions:
for example, they hold even in the seemingly inauspicious case when $\mathcal{M}$
is constant, meaning the parties are restricted to repeatedly exchanging a fixed message.

In Section \ref{sec:improvs}, we discuss other more efficient constructions that
can be used in settings where a stronger assumption --- namely, that $\mathcal{M}$
has a certain amount of min-entropy --- is acceptable. This stronger assumption, 
while not without loss of generality, might be rather benign in certain contexts 
(for example, if the messages exchanged are images).
\end{remark}

In all the protocols we consider, the symbol $\state$ is used to denote internal state
kept locally by $P_0$ and $P_1$. It is implicitly assumed that each party's state
contains an up-to-date transcript of all messages received during the protocol.
Parties may additionally keep other information in their internal state,
as a function of the local computations they perform. For $i\geq 1$,
$\state_\pr{b}{i}$ denotes the state of $P_b$ at the conclusion of exchange-round $i$.
Initial states
$\state_\pr{b}{0}=\emptyset$ are empty.

We begin with a simpler definition that only syntactically allows for the transmission of a single message
(Definition~\ref{def:subliminal}). This both serves as a warm-up to the multi-message definition presented next (Definition~\ref{def:mm-subliminal}), and will be used in its own right to prove impossibility results. See Remark~\ref{rmk:equivalence-mm} for further discussion of the relationship between these two definitions.

\begin{definition}%[Subliminal communication scheme]
\label{def:subliminal}
A \emph{subliminal communication scheme} is a two-party protocol:
\begin{displaymath}
\Pi^\E=\bigl(\Pi_\pr{0}{1}^\E,\Pi^\E_\pr{1}{1},\Pi^\E_\pr{0}{2},\Pi^\E_\pr{1}{2},
	\dots,\Pi^\E_\pr{0}{r},\Pi^\E_\pr{1}{r};\Pi^\E_{\party{1},{\sf out}}\bigr)
\end{displaymath}
	where $r\in\poly$ is the number of exchange-rounds
	and each $\Pi^\E_\pr{b}{i}$ is a PPT
	algorithm with oracle access to the algorithms of a public-key encryption scheme $\E$.
	Party $P_0$ is assumed to receive as input a message $\msg$ 
	(of at least one bit)
	that is to be conveyed to $P_1$ in an undetectable fashion. 
	The algorithms $\Pi^\E_\pr{b}{i}$ are used by $P_b$ in round $i$, respectively,
	and $\Pi^\E_{\party{1},{\sf out}}$ denotes the algorithm run by $P_1$ to
	produce an output $\msg'$ at the end of the protocol.

A subliminal communication scheme must satisfy the following syntax,
correctness and security guarantees.

\begin{itemize}
\item \textbf{Syntax.} In each exchange-round $i=1,\dots, r$:

		$P_0$ performs the following steps:
			\begin{enumerate}
				\item Sample ``innocuous message'' $m_\pr{0}{i}\leftarrow\mathcal{M}(\btau_\pr{0}{i-1})$. 
					%and set $\btau_i=m_i::\btau_{i-1}$.
				\item\label{itm:genciph0} Generate ciphertext and state $(c_\pr{0}{i},\state_\pr{0}{i})\leftarrow\Pi_\pr{0}{i}^\E(\msg,m_\pr{0}{i},\pk_1,\state_\pr{0}{i-1})$.
				\item Locally store $\state_\pr{0}{i}$ and send $c_\pr{0}{i}$ to $P_1$.
			\end{enumerate}
		Then, $P_1$ performs the following steps:\footnote{Note that the steps executed by $P_0$ and $P_1$ are entirely symmetric except in the following two aspects: first, $P_0$'s input $\msg$ is present in step \ref{itm:genciph0} but not in step \ref{itm:genciph1}; and secondly, the state $\state_\pr{1}{i-1}$ used in step \ref{itm:genciph1} contains the round-$i$ message $c_\pr{0}{i}$, whereas the state $\state_\pr{0}{i-1}$ used in step \ref{itm:genciph0} depends only on the transcript until round $i-1$.}
			\begin{enumerate}
				\item Sample ``innocuous message'' $m_\pr{1}{i}\leftarrow\mathcal{M}(\btau_\pr{1}{i-1})$. 
					%and set $\btau_i=m_i::\btau_{i-1}$.
				\item\label{itm:genciph1} Generate ciphertext and state $(c_\pr{1}{i},\state_\pr{1}{i})\leftarrow\Pi_\pr{1}{i}^\E(m_\pr{1}{i},\pk_0,\state_\pr{1}{i-1})$.
				\item Locally store $\state_\pr{1}{i}$ and send $c_\pr{1}{i}$ to $P_0$.
			\end{enumerate}
		After $r$ rounds, $P_1$ computes ${\sf msg}'=\Pi_{\party{1},{\sf
		out}}^\E(\sk_1,\state_\pr{1}{r})$ and halts.

\item\textbf{Correctness.} For any $\msg\in\{0,1\}^\kappa$, if $P_0$ and $P_1$
	play $\Pi^\E$ honestly, then $\msg'=\msg$ with probability
		$1-\negl(\kappa)$.
		The probability is taken over the key generation
		$(\pk_1,\sk_1),(\pk_2,\sk_2)\gets\E.\Gen$ and the
		randomness of the protocol algorithms, 
		as well as the message distribution $\cM$.

\item\textbf{Subliminal Indistinguishability.}
	For any semantically secure public-key encryption scheme $\E$, any
		${\sf msg}\in\{0,1\}^\kappa$ and any next-message distribution $\cM$,
		for $(\pk_1,\sk_1),(\pk_2,\sk_2)\gets\E.\Gen$, 
		the following distributions are computationally indistinguishable:

		\begin{center}
\begin{tabular}{|p{0.42\textwidth}|p{0.42\textwidth}|}
\hline
\underline{${\sf Ideal}(\pk_1,\sk_1,\pk_2,\sk_2,\cM)${\bf :}}
	& \underline{${\sf Subliminal}_\Pi({\sf msg},\pk_1,\sk_1,\pk_2,\sk_2,\cM)${\bf :}} \\
{\small for $i=1,\dots,r$:

	\quad $m_\pr{0}{i}\gets\cM(\btau_\pr{0}{i})$

	\quad $m_\pr{1}{i}\gets\cM(\btau_\pr{1}{i})$

	\quad $c_\pr{0}{i}\leftarrow\E.\Enc(\pk_{1},m_\pr{0}{i})$

	\quad $c_\pr{1}{i}\leftarrow\E.\Enc(\pk_{0},m_\pr{1}{i})$

output $\bigl(\pk_1,\sk_1,\pk_2,\sk_2;(c_\pr{b}{i})_{b\in\{0,1\},i\in[r]}\bigr)$ } &
{\small for $i=1,\dots,r$:

	\quad $m_\pr{0}{i}\gets\cM(\btau_\pr{0}{i})$

	\quad $m_\pr{1}{i}\gets\cM(\btau_\pr{1}{i})$

	\quad $(c_\pr{0}{i},\state_\pr{0}{i})\leftarrow\Pi_\pr{0}{i}^\E(\msg,m_\pr{0}{i},\pk_1,\state_\pr{0}{i-1})$

	\quad $(c_\pr{1}{i},\state_\pr{1}{i})\leftarrow\Pi_\pr{1}{i}^\E(m_\pr{1}{i},\pk_0,\state_\pr{1}{i-1})$

output $\bigl(\pk_1,\sk_1,\pk_2,\sk_2;(c_\pr{b}{i})_{b\in\{0,1\},i\in[r]}\bigr)$ }\\
\hline
\end{tabular}
		\end{center}

If the \emph{subliminal indistinguishability} requirement is satisfied only for next-message distributions $\cM$ in a restricted set $\mathbb{M}$, rather than for any $\cM$, then $\Pi$ is said to be a \emph{subliminal communication scheme for $\mathbb{M}$}.
\end{itemize}
\end{definition}

%Next, we present a slightly relaxed version of Definition \ref{def:subliminal},
%in which the indistinguishability guarantee need only hold for a specified class of message distributions $\cM$. Our main construction satisfies the stronger Definition \ref{def:subliminal}, but Definition \ref{def:subliminal-for} is useful for the more efficient constructions presented in Section \ref{sec:improvs}.

%\begin{definition}\label{def:subliminal-for}
%A \emph{subliminal communication scheme for $\mathbb{M}$} is a subliminal communication scheme satisfying a relaxed version of the ``subliminal indistinguishability'' requirement as follows. The distributions
%${\sf Real}$ and ${\sf Subliminal}$, referenced below,
%are as defined in Definition \ref{def:subliminal}.
%\begin{itemize}
%	\item\textbf{Subliminal Indistinguishability.}
%	For any semantically secure public-key encryption scheme $\E$, any
%		${\sf msg}\in\{0,1\}^\kappa$ and any message distribution 
%		$\cM\in\mathbb{M}$, for $(\pk_1,\sk_1),(\pk_2,\sk_2)\gets\E.\Gen$, 
%		the distributions ${\sf Real}(\pk_1,\sk_1,\pk_2,\sk_2,\cM)$ and
%		${\sf Subliminal}({\sf msg},\pk_1,\sk_1,\pk_2,\sk_2,\cM)$
%		are computationally indistinguishable.
%\end{itemize}
%\end{definition}

\begin{definition}
The \emph{rate} of a subliminal communication protocol $\Pi$ is defined as
	$\frac{2r}{\kappa}$, where $r$ is defined as in Definition
	\ref{def:subliminal}.\footnote{The factor of two comes from the fact that
	each exchange-round contains two messages.} This is the average number of
	bits which are subliminally communicated per ciphertext of $\E$.
	% exchanged between $P_0$ and $P_1$.
\end{definition}

For simplicity, Definition \ref{def:subliminal} presents a communication scheme in which only a single hidden message $\msg$ is transmitted. More generally, it is desirable to transmit multiple messages, and bidirectionally, and perhaps in an adaptive manner.\footnote{That is, the messages to be transmitted may become known as the protocol progresses, rather than all being known at the outset. This is the case, for example, if future messages depend on responses to previous ones.} In multi-message schemes, it may be beneficial for efficiency that the protocol have a two-phase structure where some initial preprocessing is done in the first phase, and then the second phase can thereafter be invoked many times to transmit different hidden messages.\footnote{As a concrete example: consider a simple protocol for transmitting a single encrypted message, consisting of key exchange followed by the transmission of message encrypted under the established key. When adapting this protocol to support multiple messages, it is beneficial to split the protocol into a one-time ``phase 1'' consisting of key exchange, and a ``phase 2'' encompassing the ciphertext transmission which can be invoked many times on different messages using the same phase-1 key. Such a protocol has much better amortized efficiency than simply repeating the single-message protocol many times, \emph{i.e.}, establishing a new key for each ciphertext.} This will a useful notion later in the paper, for our constructions, so we give the definition of a multi-message scheme here.

\begin{definition}\label{def:mm-subliminal}
A \emph{multi-message subliminal communication scheme} is a two-party protocol defined by a pair $(\SetupPh,\CommPh)$
where $\SetupPh$ (``Setup Phase'') and $\CommPh$ (``Communication Phase'') each define a two-party protocol. Each party outputs a state at the end of $\SetupPh$, which it uses as an input in each subsequent invocation of $\CommPh$. An execution of a multi-message subliminal communication scheme consists of an execution of $\SetupPh$ followed by one or more executions of $\CommPh$. More formally:
\begin{gather*}
\SetupPh^\E=\bigl(\Phi_\pr{0}{1}^\E,\Phi^\E_\pr{1}{1},\Phi^\E_\pr{0}{2},\Phi^\E_\pr{1}{2},
	\dots,\Phi^\E_\pr{0}{r},\Phi^\E_\pr{1}{r}\bigr) \\
\CommPh^\E=\bigl(\Xi_\pr{0}{1}^\E,\Xi^\E_\pr{1}{1},\Xi^\E_\pr{0}{2},\Xi^\E_\pr{1}{2},
	\dots,\Xi^\E_\pr{0}{r'},\Xi^\E_\pr{1}{r'};\Xi^\E_{\party{1},{\sf out}}\bigr)
\end{gather*}
	where $r,r'\in\poly$ are the number of exchange-rounds in $\Phi$ and $\Xi$ respectively.
	and where each $\Phi^\E_\pr{b}{i},\Xi^\E_\pr{b}{i}$ is a PPT
	algorithm with oracle access to the algorithms of a public-key encryption scheme $\E$. 
	The protocol must satisfy the following syntax, correctness and security guarantees.

\begin{itemize}
\item \textbf{Syntax.} In each exchange-round $i=1,\dots, r$ of $\Phi$: $P_0$ executes the following steps for $b=0$, and then $P_1$ executes the same steps for $b=1$.
	\begin{enumerate}
		\item Sample ``innocuous message'' $m_\pr{b}{i}\leftarrow\mathcal{M}(\btau_\pr{b}{i-1})$.
		\item Generate ciphertext and state 
			$(c_\pr{b}{i},\state_\pr{b}{i})\leftarrow\Phi_\pr{b}{i}^\E(m_\pr{b}{i},\pk_{1-b},\state_\pr{b}{i-1})$.
			%$(c_{i,b},\state_{i,b})\leftarrow\Phi_{i,b}^\E(\msg,m_{i,b},\pk_{1-b},\state_{i-1,b})$.
		\item Locally store $\state_\pr{b}{i}$ and send $c_\pr{b}{i}$ to $P_{1-b}$.
	\end{enumerate}
	After the completion of $\Phi$, either party may initiate $\Xi$ by sending a first message of the $\Xi$ protocol (with respect to a message $\sf msg$ to be steganographically hidden, known to the initiating party).
	\newcommand{\sender}{P_S}
	\newcommand{\receiver}{P_R}
	Let $\sender$ and $\receiver$ denote the initiating and non-initiating parties in an execution of $\Xi$, respectively.\footnote{Subscripts $S,R\in\{0,1\}$ stand for ``sender'' and ``receiver,'' respectively.}
	Let $\msg\in\{0,1\}^\kappa$ be the hidden message that $\sender$ is to transmit to $\receiver$ in an undetectable fashion during an execution of $\Xi$.

	The execution of $\Xi$ proceeds as follows over exchange-rounds $i'=1,\dots, r'$:
	\begin{itemize}
		\item $\sender$ acts as follows:
			\begin{enumerate}
				\item Sample $m_\pr{S}{r+i'}\leftarrow\mathcal{M}(\btau_\pr{S}{r+i'-1})$.
				\item Generate $(c_\pr{S}{r+i'},\state_\pr{S}{r+i'})\leftarrow\Xi_\pr{0}{i'}^\E(\msg,m_\pr{S}{r+i'},\pk_R,\state_\pr{S}{r+i'-1})$.
				\item Locally store $\state_\pr{S}{r+i'}$ and send $c_\pr{S}{r+i'}$ to $\receiver$.
			\end{enumerate}
		\item $\receiver$ acts as follows:
			\begin{enumerate} 
				\item Sample $m_\pr{R}{r+i'}\leftarrow\mathcal{M}(\btau'_\pr{R}{r+i'-1})$.
				\item Generate $(c_\pr{R}{r+i'},\state_\pr{R}{r+i'})\leftarrow\Xi_\pr{1}{i'}^\E(m_\pr{R}{r+i'},\pk_S,\state_\pr{R}{r+i'-1})$.
				\item Locally store $\state_\pr{R}{r+i'}$ and send $c_\pr{R}{r+i'}$ to $\sender$.
			\end{enumerate}
	\end{itemize}
		At the end of an execution of $\Xi$,
		$P_R$ computes ${\sf msg}'=\Xi^\E_{\party{1},{\sf
		out}}(\sk_1,\state_\pr{1}{r+r'})$.
\item\textbf{Correctness.} For any $\msg\in\{0,1\}^\kappa$, if $P_0$ and $P_1$
	execute $(\Phi,\Xi)$ honestly, then for every execution of $\Xi$, the transmitted and received messages $\msg$ and $\msg'$ are equal with overwhelming probability.
	The probability is taken over the key generation
	$(\pk_1,\sk_1),(\pk_2,\sk_2)\gets\E.\Gen$ and the
	randomness of the protocol algorithms, 
	as well as the message distribution $\cM$.

\item\textbf{Subliminal Indistinguishability.}
For any semantically secure public-key encryption scheme $\E$, any polynomial $p=p(\kappa)$, any sequence of hidden messages
		$\vec{\msg}=(\msg_i)_{i\in[p]}\in(\{0,1\}^\kappa)^p$, any sequence of bits $\vec{b}=(b_1,\dots,b_p)\in\{0,1\}^p$ and any next-message distribution $\cM$,
for $(\pk_b,\sk_b)\gets\E.\Gen$, $b\in\zo$ the following distributions are computationally indistinguishable:

		\begin{center}
	\begin{tabular}{|p{0.42\textwidth}|p{0.45\textwidth}|}
	\hline
	\underline{${\sf Ideal}(\pk_1,\sk_1,\pk_2,\sk_2,\cM)${\bf :}}
		& \underline{${\sf Subliminal}_{\Phi,\Xi}(\vec{\msg},\vec{b},\pk_1,\sk_1,\pk_2,\sk_2,\cM)${\bf :}} \\
	{\small for $i=1,\dots,r+pr'$:
	
		\quad $m_\pr{0}{i}\gets\cM(\btau_\pr{0}{i})$
	
		\quad $m_\pr{1}{i}\gets\cM(\btau_\pr{1}{i})$
	
		\quad $c_\pr{0}{i}\leftarrow\E.\Enc(\pk_{1},m_\pr{0}{i})$
	
		\quad $c_\pr{1}{i}\leftarrow\E.\Enc(\pk_{0},m_\pr{1}{i})$
	
	output:

		\quad $\bigl(\pk_1,\sk_1,\pk_2,\sk_2;(c_\pr{b}{i})_{b\in\{0,1\},i\in[r+pr']}\bigr)$ } &
	{\small for $i=1,\dots,r$:
	
		\quad $m_\pr{0}{i}\gets\cM(\btau_\pr{0}{i})$
	
		\quad $m_\pr{1}{i}\gets\cM(\btau_\pr{1}{i})$
	
		\quad $(c_\pr{0}{i},\state_\pr{0}{i})\leftarrow\Phi_\pr{0}{i}^\E(\msg,m_\pr{0}{i},\pk_1,\state_\pr{0}{i-1})$
	
		\quad $(c_\pr{1}{i},\state_\pr{1}{i})\leftarrow\Phi_\pr{1}{i}^\E(m_\pr{1}{i},\pk_0,\state_\pr{1}{i-1})$

	for $j=1,\dots,p$:

		\newcommand{\bj}{\beta}
		\newcommand{\iot}{\iota}
		\newcommand{\notbj}{\bar{\beta}}
		\quad let $\bj=b_j$ and $\notbj=1-b_j$

		\quad for $i'=1,\dots,r'$:

			\quad\quad let $\iot=r+(j-1)r'+i'$
			
			\quad\quad $m_\pr{\bj}{\iot}\gets\cM(\btau_\pr{\bj}{\iot})$
		
			\quad\quad $m_\pr{\notbj}{\iot}\gets\cM(\btau_\pr{\notbj}{\iot})$
		
			\quad\quad
			$(c_\pr{\bj}{\iot},\state_\pr{\bj}{\iot})\leftarrow\Xi_\pr{\bj}{i'}^\E(\msg,m_\pr{\bj}{\iot},\pk_{\notbj},\state_\pr{\bj}{\iot-1})$
		
			\quad\quad
			$(c_\pr{\notbj}{\iot},\state_\pr{\notbj}{\iot})\leftarrow\Xi_\pr{\notbj}{i'}^\E(m_\pr{\notbj}{\iot},\pk_{\bj},\state_\pr{\notbj}{\iot-1})$
	
	output:

		\quad
		$\bigl(\pk_1,\sk_1,\pk_2,\sk_2;(c_\pr{b}{i})_{b\in\{0,1\},i\in[r+pr']}\bigr)$ }\\
	\hline
	\end{tabular}
		\end{center}

If the \emph{subliminal indistinguishability} requirement is satisfied only for
$M$ in a restricted set $\mathbb{M}$, rather than for any $\cM$, then
$(\Phi,\Xi)$ is said to be a \emph{multi-message subliminal communication
scheme for $\mathbb{M}$}.

\end{itemize}
\end{definition}

\begin{definition}
	The \emph{asymptotic rate} of a multi-message subliminal communication
	protocol $(\Phi,\Xi)$ is defined as $\frac{\kappa}{2r'}$, where $r'$ is
	defined as in Definition \ref{def:mm-subliminal}.  The asymptotic rate is
	the average number of bits which are subliminally communicated per
	ciphertext exchanged between $P_0$ and $P_1$ after the one-time setup phase
	is completed.
\end{definition}

\begin{definition}
	The \emph{setup cost} of a multi-message subliminal communication protocol
	$(\Phi,\Xi)$ is defined as $r$, \emph{i.e.}, the number of rounds in
	$\Phi$.  The setup cost is the number of ciphertexts which must be sent
	back and forth between $P_0$ and $P_1$ in order to complete the setup
	phase.
\end{definition}

\begin{remark}\label{rmk:equivalence-mm}
	Definition \ref{def:mm-subliminal} is \emph{equivalent} to Definition \ref{def:subliminal} in the sense that the existence of any single-message scheme trivially implies a multi-message scheme and vice versa. We present Definition \ref{def:mm-subliminal} as it will be useful for presenting and analyzing asymptotic efficiency of our constructions, but note that this equivalence means that the simpler Definition \ref{def:subliminal} suffices in the context of impossibility (or possibility) results, such as that given in Section~\ref{sec:imposs}.
\end{remark}

\section{Impossibility Results}
\label{sec:imposs}

\subsection{Locally Decodable Subliminal Communication Schemes}
\label{sec:imposs-subl}

%In this work, we focus on the case  where the cover distribution consists of
%ciphertexts of a semantically secure encryption scheme. We refer to such
%schemes as \emph{subliminal communication schemes}, of which a formal
%definition is given in Section~\ref{sec:subl}. Despite that in a subliminal
%communication scheme, the cover distribution is more structured than an
%arbitrary cover distribution (as in the case of a general steganographic
%communication scheme), designing such schemes is still a challenging task. 

A first attempt at achieving subliminal communication might
consider schemes with the following
natural property: the receiving party $P_1$ extracts hidden bits \emph{one
ciphertext at a time}, by the application of a single (possibly randomized)
decoding function. We refer to such schemes as \emph{locally decodable} and our
next impossibility theorem shows that non-trivial locally decodable schemes do
not exist if the encryption scheme $\E$ is chosen
adversarially.

\begin{theorem}
	\label{thm:impossibility-1}
	For any locally decodable protocol $\Pi$ satisfying the syntax of a single-message\footnote{Remark~\ref{rmk:equivalence-mm} discusses the sufficiency of proving impossibility for single-message schemes.} 
	subliminal communication scheme, there exists
	a semantically secure public-key encryption scheme $\E$ dependent on the public randomness of $\Pi$, such that $\E$ violates the correctness condition of Definition \ref{def:subliminal}. Therefore, no locally decodable protocol $\Pi$ is a subliminal communication scheme.
\end{theorem}
\begin{proof}
	Let us consider a locally decodable scheme such as in the statement of the
	theorem, and let us denote by $\Pi_{2, \sf out}:\zo^n\times\zo^m\to\zo$ the
	decoding function of the scheme where the second input consists of random
	bits (the public randomness) and the first input is a ciphertext $c$. Since
	we allow the encryption scheme to depend on the public randomness of the
	subliminal scheme, define the partial function $f_r(c) = \Pi_{1, \sf
	out}(c, r)$. $f_r$ is now a deterministic function of the ciphertext and we
	conclude the proof by constructing an encryption scheme which biases the
	output of $f_r$ arbitrarily close to a constant bit. This is
	a contradiction, since by correctness and subliminal indistinguishability,
	$f_r(c)$ should have negligible bias when subliminally communicating
	a uniformly random message $\msg\gets\{m_1,m_2\}$.

	Let $\E,$ be a semantically secure encryption scheme with ciphertext space
	$\Ciph=\zo^n$ and message space $\Msg$. Without loss of generality we
	assume that for at least half the messages $m\in\Msg$, we have
	$\Pr[f_r\big(\E.\Enc(pk,m)\big)=1]\geq\frac{1}{2}$ (otherwise we can just
	replace 1 by
	0 in the construction below). We now define the encryption scheme $\E'$
	  which is identical to $\E$ except for $\E'.\Enc$ which on input $(pk, m)$
	  runs as follows for some constant $t$.

		\begin{enumerate}
		 \item Repeat at most $t$ times:
			\begin{enumerate}
				\item Sample encryption $c\gets\E.\Enc(pk, m)$.
				\item If $f_r(c)=1$, exit the loop; otherwise, continue.
			\end{enumerate}
		\item Output $c$.
		\end{enumerate}

	It is clear that $\E'$ is also semantically secure: oracle access to
	$\E'.\Enc$ can be simulated with oracle access to $\E.\Enc$, so
	a distinguisher which breaks the semantic security of $\E'$ can also be
	used to break the semantic security of $\E$. Finally, for a message $m$
	such that $\Pr[f_r(C_m)=1]\geq\frac{1}{2}$, by definition of
	$\E'.\Enc$, it holds that 
	\begin{displaymath}
		\Pr[f_r(\E'.\Enc(PK, m))=1]\geq1-\frac{1}{2^t}\ .
	\end{displaymath}
	This shows that the output of $f_r$ can be arbitrarily biased and
	concludes the proof.
\end{proof}

\begin{remark}
	The essence of the above theorem is the impossibility of deterministic extraction: no single deterministic function can deterministically extract from ciphertexts of arbitrary encryption schemes. The way to bypass this impossibility is to have the extractor depend on the encryption scheme. Note that multiple-source extraction, which is used in our constructions in the subsequent sections, implicitly do depend on the underlying encryption scheme, since the additional sources of input depend on the encryption scheme and thus can be thought of as ``auxiliary input'' that is specific to the encryption scheme at hand.
\end{remark}

%steganographic stuff starts here
\subsection{Steganography for Adversarial Cover Distributions}
\label{sec:imposs-stego}

Our second impossibility result concerns a much more general class of
communication schemes, which we call \emph{steganographic communication schemes}.
Subliminal communication schemes, as well as the existing notions of public-key steganography and steganographic key exchange from the steganography literature, are instantiations of the more general definition of a (multi-message) steganographic communication scheme.
To our knowledge, the general notion of a steganographic communication scheme has not been formalized in this way in prior work. In the context of this work, the general definition is helpful for proving broad impossibilities across multiple types of steganographic schemes.

As mentioned in the introduction, a limitation of all existing results in the
steganographic literature, to our knowledge, is that they assume that the \emph{cover distribution} --- \emph{i.e.}, the distribution of innocuous objects in which steganographic communication is to be embedded --- is
fixed \emph{a priori}. 
In particular, the cover distribution is assumed not to depend on the description
of the steganographic communication scheme. The impossibility result given in Section~\ref{sec:imposs-subl} is an example illustrative of the power of adversarially choosing the cover distribution: Theorem \ref{thm:impossibility-1} says that by choosing the encryption scheme $\E$ to depend on a given subliminal communication scheme, an adversary can rule out the possibility of any hidden communication at all.

Our next impossibility result (Theorem~\ref{thm:impossibility-2})
shows that if the cover distribution is chosen adversarially, then non-trivial 
steganographic communication is impossible.

%Steganographic schemes for a variety of cryptographic tasks have been constructed in the literature beginning with the work of Simmons~\cite{Sim83}.  All steganographic schemes are defined with respect to a cover distribution $\cC$, and the parties embed their hidden message inside ``innocent looking'' cover messages distributed according to $\cC$. 

\begin{theorem}
\label{thm:impossibility-2}
Let $\Pi$ be a steganographic communication scheme. Then for any $k\in\NN$,
	there exists a cover distribution $\cC$ of conditional min-entropy $k$ such
	that the steganographic indistinguishability of $\Pi$ does not hold for
	more than one message.
\end{theorem}

In Appendix~\ref{a:imposs}, we give the formal definition of a \emph{steganographic communication scheme}, along with the proof of Theorem~\ref{thm:impossibility-2}. We have elected to present these in the appendix as the definition introduces a set of new notation only used for the corresponding impossibility result, and both the definition and the impossibility result are somewhat tangential to the main results of this work, whose focus is on subliminal communication schemes.

\section{Construction of the Subliminal Scheme}
\label{sec:const}
	
The goal of this section is to establish the following theorem, which
states that our construction $(\OurPhi,\OurXi)$ is a subliminal communication scheme when instantiated with a pseudorandom key-exchange protocol (such as Diffie-Hellman).

\begin{theorem}
\label{thm:subl}
	% Lower down in this section there is a copy of this theorem
	The protocol 
	$(\OurPhi,\OurXi)$ given in Definition~\ref{def:full-protocol},
	when instantiated with a pseudorandom key-exchange protocol $\Lambda$,
	is a multi-message subliminal communication scheme.
\end{theorem}

The detailed description and proofs of security and correctness of our scheme
can be found in the following subsections.  Our construction makes no
assumption on the message distribution $\cM$ and in particular holds when the
exchanged plaintexts (of the adversarially mandated encryption scheme $\E$) are a fixed, adversarially chosen sequence of messages.
An informal outline of the construction is given next.

\begin{definition}\label{def:outline} Outline of the construction.
\begin{enumerate}
\item {\bf Setup Phase} $\OurPhi$
	\begin{enumerate}
		\item\label{itm:seed} 
			A $\widetilde O(\log \kappa)$-bit string $S$ is established between
			$P_0$ and $P_1$ by extracting randomness from pairs of consecutive
			ciphertexts.  (\emph{Protocol overview in Section
			\ref{sec:establish-seed}.})

		\item\label{itm:KE}
			Let $\Ext$ be a strong seeded extractor, and let $S$ serve as its seed. By rejection-sampling
			ciphertexts $c$ until $\Ext_S(c)=\str$, either party can
			embed a random string $\str$ \emph{of their choice} in the
			conversation. (\emph{Protocol overview in Section \ref{sec:overt}.})
			By embedding in this manner 
			the messages of a pseudorandom key-exchange protocol, 
			both parties establish a shared secret $\sk^*$.\footnote{Note that the random string $\str$ is known to an eavesdropper who has knowledge of the seed $S$. Nonetheless, (1) the established secret $\sk^*$ is unknown to the eavesdropper by the security of the key-exchange protocol and (2) the transcript is indistinguishable to the eavesdropper from one in which no key exchange occurred at all, due to the pseudorandomness of the key-exchange messages. \label{ft:overt}}
	\end{enumerate}

\item {\bf Communication Phase} $\OurXi$ \label{itm:comm} \\
	Both parties can now communicate arbitrary messages of their choice by
	\emph{(1)} encrypting them using a pseudorandom secret-key encryption scheme $\SKE$ using $\sk^*$ as the secret key, and \emph{(2)} embedding
	the ciphertexts of $\SKE$ using the rejection-sampling technique described in Step \ref{itm:KE}.\footnote{Again, an eavesdropper could know the $\SKE$ ciphertexts exchanged, if he knew the seed $S$, but could not distinguish the $\SKE$ ciphertexts from truly random strings, and thus could not tell whether any subliminal communication was occuring at all. \emph{Cf.} footnote \ref{ft:overt}.}
	(\emph{Detailed protocol in Section \ref{sec:full-protocol}.})
\end{enumerate}
The full protocol is given, and proven to be a subliminal communication scheme, in Section~\ref{sec:full-protocol}.
\end{definition}

\subsection{Establishing a Shared Seed}
\label{sec:establish-seed}

In this section, we give a protocol which allows $P_0$ and $P_1$ to establish
a random public parameter which will be used in subsequent phases of our subliminal
scheme. As such, this can be thought of as drawing a subliminal scheme at
random from a family of subliminal schemes. The parameter is public in the
sense that anyone eavesdropping on the channel between $P_0$ and $P_1$ gains
knowledge of it. A crucial point is that the random draw occurs \emph{after}
the adversarial encryption scheme $\E$ is fixed, thus bypassing the
impossibility results of Section~\ref{sec:imposs}.

Our strategy is simple: extract randomness from pairs of ciphertexts. Since
the extractor does not receive the key, semantic security holds with respect 
to the extractor: a pair of ciphertexts for two arbitrary messages is
indistinguishable from two encryptions of a fixed message; thus, a 
same-source extractor suffices for our purposes (see
Lemma~\ref{lemma:2ext-comp}). 
Even though semantic security guarantees only $\omega(\log \kappa)$ min-entropy
of ciphertexts (see Lemma \ref{lemma:cond-min-ent}), 
we will be able to make use of the ``greater-than'' extractor 
(Definition \ref{def:gt}) applied to pairs of ciphertexts, 
and obtain Theorem~\ref{thm:ext}.

\begin{definition}\label{def:gt}
	The \emph{greater-than extractor} $\Gt$ is defined by $\Gt(x, y) \defeq \bo[x\geq y]$.
\end{definition}

\begin{lemma}[Ciphertexts have super-logarithmic min-entropy]
	\label{lemma:cond-min-ent}
	Let $\PKE$ be a semantically secure encryption scheme. Then there exists
	a negligible function $\eps$ such that for all $\kappa\in\NN$,
	$m\in\Msg_{\kappa}$, writing $C^{pk}_m\sim\PKE.\Enc(pk,m)$:
	\begin{displaymath}
		\Pr\left[(pk, sk)\gets\PKE.\Gen(1^\kappa):
		\Me\big(C_m^{pk}\big)\geq\log\frac{1}{\eps(\kappa)}\right]
		\geq 1-\eps(\kappa)\;.
	\end{displaymath}
\end{lemma}
\begin{proof}
	In Appendix~\ref{a:me}.
\end{proof}

Given that ciphertexts of semantically secure encryption schemes have
min-entropy $\omega(\log \kappa)$, we will consider extractors which have
negligible bias on such sources. This motivates the following definition.

\begin{definition}
	Let $\TwoExt: \zo^{n} \times \zo^{m} \to\zo^{\ell}$ be a two-source
	extractor, we say that $\TwoExt$ is an \emph{extractor for
	super-logarithmic min-entropy} if $\TwoExt$ is a $(d\log \kappa, d\log
	\kappa, \frac{1}{\kappa^d})$ extractor for any $d\in\NN$. In
	particular, for any negligible function $\eps$, there exists a negligible
	function $\eps'$ such that $\TwoExt$ is a $(\log\frac{1}{\eps},
	\log\frac{1}{\eps}, \eps')$
	extractor.
\end{definition}

The following lemma shows that the output of a same-source extractor for
super-logarithmic min-entropy on two ciphertexts is statistically
indistinguishable from uniform, even in the presence of the key.

\begin{lemma}
	\label{lemma:2ext-comp}
	Let $\PKE$ be a semantically secure encryption scheme with ciphertext
	length $n$, and let $\TwoExt:\zo^n\times\zo^n\to\zo^\ell$ be a same-source
	extractor for super-logarithmic min-entropy with
	$\ell(\kappa)=O(\log\kappa)$, then there exists a negligible function
	$\eps$ such that, for any $\kappa\in\NN$, $(m_0, m_1)\in\Msg_{\kappa}^2$,
	writing $(PK, SK)\sim\PKE.\Gen(1^\kappa)$, $C^{pk}_i\sim\PKE.\Enc(pk,
	m_i)$, $i\in\zo$:
	\begin{displaymath}
		\sdist{\big(PK, SK,
		\TwoExt(C^{PK}_0, C^{PK}_1)\big)
		-(PK, SK, U_{\ell(\kappa)})}
		\leq \eps(\kappa)\;.
	\end{displaymath}
\end{lemma}

\begin{proof}
We will prove that for any polynomial $p$ and for large enough $\kappa$:
	\begin{displaymath}
		\Pr\left[
\begin{array}{lll}
	(pk, sk)\gets\PKE.\Gen(1^\kappa)
&:&
	\sdist{\TwoExt(C_0^{pk}, C_1^{pk}) - U_{\ell(\kappa)}}
\end{array}
		\leq \frac{1}{p(\kappa)}\right]
		\geq 1-\frac{1}{p(\kappa)}\;.
	\end{displaymath}
Assume by contradiction that there exists $p\in\poly(\kappa)$ and an infinite
	set $I\subseteq\NN$ such that for $\kappa\in I$:
	\begin{equation}
		\label{eq:2ext}
		\Pr\left[
\begin{array}{lll}
	(pk, sk)\gets\PKE.\Gen(1^\kappa)
&:&
\sdist{\TwoExt(C_0^{pk}, C_1^{pk}) - U_{\ell(\kappa)}} 
\end{array}
		\geq \frac{1}{p(\kappa)}
\right]
		\geq \frac{1}{p(\kappa)}\;.
	\end{equation}

	We now construct an adversary $\sD$ distinguishing between
	$(PK, C_0^{PK})$ and $(PK, C_1^{PK})$ with non-negligible
	advantage. On input $(pk, c)$, $\sD$ runs as follows:
	\begin{enumerate}
		\item Sample two encryptions  of $m_0$: $c_0,
			c_0'\stackrel{iid}{\gets} \PKE.\Enc(pk, m_0)$,
			and $c_1\gets\PKE.\Enc(pk, m_1)$.
		\item If $\TwoExt(c_0, c_1) = \TwoExt(c_0', c)$ output 1, otherwise
			output 0.
	\end{enumerate}

	First, note that on input $(PK, C_1^{PK})$, $\sD$ outputs 1 iff a collision
	occurs at step 2. By \eqref{eq:2ext}, with probability at least
	$\frac{1}{p}$ over the draw of the $(pk, sk)$, a collision occurs
	with probability at least $\frac{1}{2^\ell}+\frac{1}{2^\ell p^2}$.
	Otherwise, a collision occurs with probability at least $\frac{1}{2^\ell}$.
	Overall, for $\kappa\in I$:
	\begin{equation}
		\label{eq:2ext-foo}
		\Pr[\sD(PK, C_1^{PK})=1]\geq \frac{1}{2^\ell}
		+\frac{1}{2^\ell p^3}\;.
	\end{equation}

	By Lemma~\ref{lemma:cond-min-ent}, after conditioning on the event that
	$\Me(C_0^{pk})\geq\log q(\kappa)$, the guarantee of $\TwoExt$ applies to
	a pair of independent encryptions of $m_0$ under $pk$ and we obtain, for
	large enough $\kappa$:
	\begin{displaymath}
		\Pr\left[
\begin{array}{lll}
	(pk, sk)\gets\PKE.\Gen(1^\kappa)
&:&
	\sdist{\TwoExt(C_0^{pk}, C_0^{'pk}) - U_{\ell(\kappa)}}
\end{array}
		\leq \frac{1}{q(\kappa)}\right]
		\geq 1-\frac{1}{q(\kappa)}\;,
	\end{displaymath}
	This implies that for large enough $\kappa$:
	\begin{equation}
		\label{eq:2ext-bar}
		\Pr[\sD(PK, C_0^{PK})=1]\leq \frac{1}{2^\ell}+\frac{2}{q}\;.
	\end{equation}

	Together, \eqref{eq:2ext-foo} and \eqref{eq:2ext-bar} imply, after
	choosing $q= 2^{\ell+2}p^3\in\poly$, that for large enough
	$\kappa\in I$:
	\begin{displaymath}
		\Pr[\sD(PK, C_0^{PK})=1]
		-
		\Pr[\sD(PK, C_1^{PK})=1]
		\geq
		\frac{2}{q}\;.
	\end{displaymath}
	This contradicts the security of $\PKE$ and concludes the proof.
\end{proof}

Finally, we observe that the ``greater-than'' extractor is a same-source
extractor for super-logarithmic min-entropy (Lemma \ref{lemma:gt}). To the best
of our knowledge, this is a folklore fact which is for example mentioned in
\cite{BIW04}.

\begin{lemma}
\label{lemma:gt}
For any $k\leq n$, $\Gt$ is a $(k, \frac{1}{2^k})$ same-source extractor.
\end{lemma}
\begin{proof}
	In Appendix~\ref{a:gt}.
\end{proof}

We now conclude this section with a full description of our method for
establishing the public parameter $S$ introduced in Step \ref{itm:seed}.

\begin{theorem}
	\label{thm:ext}
	Let $\E$ be a semantically secure public-key encryption scheme
	and let $\rho\in\poly$. Define random variables as follows.
	\begin{itemize}
		\item For $b\in\zo$, let $K_b=(PK_b, SK_b)=\E.\Gen(1^\kappa)$.
		\item For $b\in\zo$ and $i\in[2\rho]$, let $C_\pr{b}{i} = \E.\Enc(PK_{1-b}, m_\pr{b}{i})$
			representing the
			ciphertexts exchanged between $P_0$ and $P_1$ during $2\rho$ exchange-rounds.
		\ifsubmission
		\item Let $S = \big(\Gt(C_\pr{0}{1},C_\pr{0}{2}),\Gt(C_\pr{1}{1},C_\pr{1}{2}),\dots, \Gt(C_\pr{1}{2\rho-1},C_\pr{1}{2\rho})\big)$.
		\else
		\item Let $S = \big(\Gt(C_\pr{0}{1},C_\pr{0}{2}),\Gt(C_\pr{1}{1},C_\pr{1}{2}),\dots, \Gt(C_\pr{0}{2\rho-1},C_\pr{0}{2\rho}),\Gt(C_\pr{1}{2\rho-1},C_\pr{1}{2\rho})\big)$.
			\fi
	\end{itemize}
	There exists a negligible function $\eps$ such that:
	\begin{displaymath}
		\sdist{(K_0,K_1, S)
		- (K_0, K_1, U_{2\rho})}\leq \eps\;.
	\end{displaymath}
\end{theorem}

\begin{proof}
	Writing $S=(S_1, S_1',\dots, S_{\rho}, S'_{\rho})$, we have:
	\begin{align*}
		\sdist{(K_0,K_1, S)
		- (K_0, K_1, U_{2\rho})}
		\leq&\sum_{i=1}^\rho\big(\sdist{(K_0,K_1, _i) - (K_0, K_1, U_1)}\\
		&\quad\quad + \sdist{(K_0,K_1, S_i') - (K_0, K_1, S_1)}\big)\leq 2\rho\eps\;,
	\end{align*}
	where the first inequality follows by independence of the ciphertexts
	conditioned on the keys, and the second inequality follows by
	Lemma~\ref{lemma:2ext-comp}.
\end{proof}

\begin{remark}
	In the construction of Theorem~\ref{thm:ext}, the ciphertexts exchanged
	between $P_0$ and $P_1$ are sent without any modification, so subliminal
	indistinguishability clearly holds at this point.
\end{remark}

\subsection{Embedding Random Strings}
\label{sec:overt}

In this section, we assume that both parties have access to a public parameter
$S$ and construct a protocol which allows for embedding of
uniformly random strings into ciphertexts of an adversarially chosen encryption scheme $\E$,
as required by Steps~\ref{itm:KE} and \ref{itm:comm} of the construction outline (Definition~\ref{def:outline}). The security guarantee is
that for a uniformly random parameter $S$ and uniformly random strings to be embedded, the
ciphertexts of $\E$ with embedded random strings are
indistinguishable from ciphertexts of $\E$ produced by direct application of $\E.\Enc$, even to an adversary who knows the decryption keys of $\E$. 
This can be thought of as a relaxation of subliminal
indistinguishability (Definition \ref{def:subliminal})
where the two main differences are that \emph{(1)} the
parties have shared knowledge of a random seed, and \emph{(2)} indistinguishability only holds when
embedding a \emph{random} string, rather than for arbitrary strings. 
We first present a construction to embed
logarithmically many random bits (Theorem~\ref{thm:ext-scheme}) and then show
how to sequentially compose it to embed arbitrarily polynomially many random bits
(Theorem~\ref{lem:comp}). These constructions rely on a strong
seeded extractor that can extract logarithmically many bits from sources of
super-logarithmic min-entropy. Almost universal hashing is a simple
such extractor, as stated in Proposition~\ref{prop:auh}.

\begin{theorem}
	\label{thm:ext-scheme}
	Let $\Ext:\zo^d\times\zo^n\to\zo^v$ be a strong seeded extractor for
	super-logarithmic min-entropy with $v=O(\log \kappa)$, and let $\E$ be
	a semantically secure encryption scheme with ciphertext space $\cC=\zo^n$.
	%Define the rejection sampler $\Sigma^{\E,S}$ with public parameter $S$ (a $d$-bit seed).
	Let $\Sigma^{\E,S}$ be defined as in Algorithm~\ref{algo:rs}, then the following guarantees hold:

	\begin{algorithm}[ht!]
		\caption{Rejection sampler $\Sigma^{\E,S}$}
		\label{algo:rs}
		\textsc{Public parameter:} $S$ (a $d$-bit seed). \\
		\textsc{Input:} $(\str,m,\pk)$ where $\str$ is the string to be embedded.
		\begin{enumerate}
			\item\label{itm:encrypt2} Generate encryption $c\gets\E.\Enc(\pk,m)$.
			\item If $\Ext(r, c)=\str$, then output $c$. Else, go back to step
				\ref{itm:encrypt2}.
		\end{enumerate}
	\end{algorithm}

	\begin{enumerate}
		\item \emph{Correctness:} for any $S\in\zo^d$ and $\str\in\zo^v$, if $c=\Sigma^{\E,
			S}(\str, m, \pk)$, and $\str'=\Ext(S, c)$, then $\str'=\str$.
		\item \emph{Security:} there exists a negligible function $\eps$ such that
		writing $ (PK, SK)=\E.\Gen(1^\kappa)$, $C=\E.\Enc(PK, m)$ and
		$C'=\Sigma^{\E, U_d}(U_v,m,PK)$, the following holds:
	\begin{displaymath}
		\sdist{(PK, SK, U_d, C) -(PK, SK, U_d, C')}
		\leq \eps(\kappa)\;.
	\end{displaymath}
	\end{enumerate}
\end{theorem}

\begin{proof}
	Define $C''=\E.\Enc(PK_2,m)$, an encryption of $m$ independent of $C$.
	By definition of rejection sampling, $C \sim \Sigma^{\E,S}\big(\Ext(S,
	C''),m, PK_2\big)$. Since $\Ext$ is a strong extractor for
	super-logarithmic min-entropy, and since $C$ has super-logarithmic
	min-entropy, there exists a negligible function $\eps$ such that:
	\begin{displaymath}
		\sdist{\big(PK_2, SK_2, U_d, \Ext(U_d, C'')\big)
		-(PK_2, SK_2, U_d, U_v)}
		\leq \eps(\kappa)\;.
	\end{displaymath}
	The statistical distance can only decrease by applying $\Sigma^{\E}$ on
	both sides, hence:
	\begin{displaymath}
		\sdist{(PK_2, SK_2, U_d, C)
		-(PK_2, SK_2, U_d, C')}
		\leq \eps(\kappa)\;.
	\end{displaymath}
	which proves the security guarantee. Correctness is immediate.
\end{proof}

\begin{remark}\label{rmk:not_inc}
Rejection sampling is a simple and natural approach that has
	been used by prior work in the steganographic literature, such as \cite{BC05}.
	Despite the shared use of this common technique, our construction is more
	different from prior art than it might seem at first glance. The novelty of
	our construction arises from the challenges of working in a model with
	a stronger adversary who can choose the distribution of ciphertexts
	(\emph{i.e.}, the adversary gets to choose the public-key
	encryption scheme $\E$).
	We manage to bypass the impossibilities outlined in
	Section \ref{sec:imposs} notwithstanding this stronger adversarial model, and in contrast to prior work, construct a protocol to established a shared seed from scratch, rather than simply assuming that one has been established in advance.
	%by leveraging the semantic security of the ciphertexts of $\E$ in order to be able to use the techniques of two-source
	%extraction (which would not work for arbitrary cover distributions) and
	%first establishing a shared parameter.
\end{remark}

We now sequentially compose 
%the construction of
Theorem~\ref{thm:ext-scheme} to embed longer strings.

\begin{theorem}
\label{lem:comp}
	Let $\Sigma$ be the rejection sampler defined in
	Algorithm~\ref{algo:rs}. Let $\ell\in\poly$ and $U_{\ell}$ be a uniformly random
	message of $\ell$ bits. For $v\leq\ell$, we write $U_{\ell} = U_{\ell, 1}||\dots ||
	U_{\ell, \nu}$ where $U_{\ell, i}$ is a block of $v$ bits from $U_\ell$ and
	$\nu = \frac{\ell}{v}$. Given cover
	messages $m_1, \dots, m_\ell$, define $(PK, SK)=\E.\Gen(1^\kappa)$, $C_i'
	= \Sigma^{\E, U_d}(U_{\ell, i}, m_i, \pk)$, $C_i = \E.\Enc(PK, m_i)$,
	then there exists a negigible function $\eps$ such that:
	\begin{displaymath}
		\sdist{(PK, SK, U_d, (C_i)_{i\in[\nu]}) -(PK, SK, U_d,
		(C_i')_{i\in[\nu]})} \leq \eps(\kappa)\;.
	\end{displaymath}
\end{theorem}

\begin{proof}
	Define $K=(PK, SK)$, then:
	\begin{displaymath}
		\sdist{\big(K, U_d, (C_i)_{i\in[\nu]}\big)
		-\big(K, U_d, (C_i')_{i\in[\nu]}\big)}
		\leq\sum_{i=1}^\nu\sdist{(K, C_i) - (K, C_i')} \leq \nu\eps\;,
	\end{displaymath}
	where the first inequality is by independence of the sequences
	$(C_i)_{i\in[\nu]}$ and $(C_i')_{i\in[\nu]}$ conditioned on the keys, and
	the second inequality is by Theorem~\ref{thm:ext-scheme}.
\end{proof}

Finally, we observe that almost universal hashing is a strong seeded extractor
for super-logarithmic min-entropy which has negligible error when the output
length is $O(\log(\kappa))$ (Proposition \ref{prop:auh}). This exactly
satisfies the requirement of Theorem \ref{thm:ext-scheme}. Moreover, the seed
legnth of this extractor is only super-logarithmic, meaning that the seed can
be established, in Step \ref{itm:seed} of the Setup Phase (Definition
\ref{def:outline}), in $\widetilde O(\log\kappa)$ many exchange-rounds of
communication.

\begin{proposition}
	\label{prop:auh}
	Let $\delta$ be a negligible function and let $\cH$ be a family of
	$\delta$-almost pairwise independent hash functions mapping $\zo^n$ to
	$\zo^{c\cdot \log n}$, then the extractor
	$\Ext:\cH\times\zo^n\to\zo^{c\cdot\log n}$ defined by $\Ext(h, x) = h(x)$
	is a strong seeded extractor for super-logarithmic min-entropy.
	Furthermore, there exists an explicit  family $\cH$ of such hash functions,
	such that sampling uniformly from $\cH$ requires $O\big(c\cdot\log
	n + \log\frac{1}{\delta}\big)$ bits.
\end{proposition}
\begin{proof}
	See \cite{SZ94}.
\end{proof}

\subsection{Full Protocol $(\OurPhi,\OurXi)$}
\label{sec:full-protocol}

First, we establish some notation for the syntax of a key-exchange protocol.

\begin{definition}[Key-exchange protocol syntax]
	A key-exchange protocol is a two-party protocol defined by
	$$\Lambda = (\Lambda_\pr{0}{1},\Lambda_\pr{1}{1},\Lambda_\pr{0}{2},\Lambda_\pr{1}{2},\dots,\Lambda_\pr{0}{k},\Lambda_\pr{1}{k},\Lambda_{\party{0},{\sf out}},\Lambda_{\party{1},{\sf out}})\ .$$ 
	We assume $k$ simultaneous communication rounds,
	where $\Lambda_\pr{b}{i}$ represents the computation 
	performed by $P_b$ in the $i$th round.
	The parties are stateful 
	and their state is implicitly updated at each round 
	to contain the transcript so far and 
	any local randomness generated so far.
	Each $\Lambda_\pr{b}{i}$ takes as input the transcript
	up to round $i-1$ and the state of $P_b$,
	 and outputs a message 
	$\lambda_\pr{b}{i}$ to be sent in the $i$th round.
	For notational simplicity, we write explicitly only the first input to $\Lambda_\pr{b}{i}$, 
	and leave the second input (\emph{i.e.}, the state) implicit.
	$\Lambda_{\party{0},{\sf out}},\Lambda_{\party{1},{\sf out}}$
	are run by $P_0,P_1$ respectively to compute the shared secret at the conclusion of the protocol.
\end{definition}

Next, we give the full construction of $(\OurPhi,\OurXi)$ following the outline in
Definition~\ref{def:outline}.

\begin{definition} $(\OurPhi,\OurXi)$ is parametrized by the following.
	\begin{itemize}
		\item $\Ext:\zo^d\times\zo^n\to\zo^{v}$, a strong seeded extractor.
		\item $\Lambda$, a pseudorandom key-exchange protocol with
		$\ell$-bit messages.\footnote{In presenting our construction $(\OurPhi,\OurXi)$, we do not denote the state of parties w.r.t. the key-exchange protocol $\Lambda$ by a separate variable, but assume that it is part of the state $\state_\pr{b}{i}$ of the overall protocol.}
		\item $\SKE$, a pseudorandom secret key encryption scheme with
			$\xi$-bit ciphertexts.
	\end{itemize}
	We define each phase of our construction in turn.
	\label{def:full-protocol}
	\begin{enumerate}\begin{linenumbers}
		\item \textbf{Setup Phase} $\OurPhi$
			\begin{enumerate}
				\item \textbf{Establishing a $d$-bit shared seed}
					\begin{itemize}
						\item For $b\in\zo$ and $i\in\{1,\dots,d\}$,
							$\OurPhi_\pr{b}{i}(m_\pr{b}{i},\pk_{1-b},\state_\pr{b}{i-1})$
							outputs a ciphertext
							$c_\pr{b}{i}=\E.\Enc(\pk_{1-b}, m_\pr{b}{i})$ and sets the
							updated state $\state_\pr{b}{i}$ to be the
							transcript of all protocol messages sent and
							received so far.
						\item At the conclusion of the $d$ exchange-rounds, each party updates his state to contain the seed $S$ which is defined by
						\begin{equation*}\label{eqn:seed}
							\ifsubmission
							S= \big(\Gt(c_\pr{0}{1},c_\pr{0}{2}),\Gt(c_\pr{1}{1},c_\pr{1}{2}),\dots,\Gt(c_\pr{1}{d-1},c_\pr{1}{d})\big)\;.
							\else
							S= \big(\Gt(c_\pr{0}{1},c_\pr{0}{2}),\Gt(c_\pr{1}{1},c_\pr{1}{2}),\dots, \Gt(c_\pr{0}{d-1},c_\pr{0}{d}),\Gt(c_\pr{1}{d-1},c_\pr{1}{d})\big)\;.
							\fi
						\end{equation*}
						This seed $S$ is assumed to be accessible in all future states throughout both phases during the remainder of the protocol.
					\end{itemize}
				\item \textbf{Subliminal key exchange} \\
					Let $\nu \defeq \frac{\ell}{v}$.
					Subliminal key exchange occurs over $k\cdot \nu$ exchange-rounds.

					\begin{itemize}
						\item For $j\in\{1,\dots, k\}$ and $b\in\zo$:
							\begin{itemize}
								\item $P_b$ retrieves from his state the
									key-exchange transcript so far
									$(\lambda_\pr{b}{j'})_{b\in\zo, j'<j}$.
								\item \linelabel{ln:KE-msg} $P_b$ 
									computes the next key-exchange message
									$$\lambda_\pr{b}{j}
									\gets\Lambda_\pr{b}{j}\big((\lambda_\pr{b}{j'})_{b\in\zo,
									j'<j}\big)\;.$$
								\item $P_b$ breaks $\lambda_\pr{b}{j}$ into
									$v$-bit blocks
									$\lambda_\pr{b}{j}
									= \lambda^1_\pr{b}{j}||\dots||\lambda^\nu_\pr{b}{j}$.
								\item The $\nu$ blocks are transmitted sequentially as follows. For $\iota\in\{1,\dots,\nu\}$:
									\begin{itemize}
										\item[] Let $i=d+(j-1)\nu+\iota$. \\
											$\OurPhi_\pr{b}{i}(m_\pr{b}{i},\pk_{1-b},\state_\pr{b}{i-1})$
											outputs $c_\pr{b}{i}\gets\Sigma^{\E,
											S}(\lambda_\pr{b}{j}^\iota,
											m_\pr{b}{i},\pk_{1-b})$ and sets
											the updated state
											$\state_\pr{b}{i}$ to contain the
											transcript of all protocol messages
											sent and received so far.
									\end{itemize}

								\item At the conclusion of the $\iota$
									exchange-rounds, each party $b\in\zo$
									updates his state to contain the $j$th
									key-exchange message $\lambda_\pr{1-b}{j}$
									computed as follows:
									\begin{displaymath}
										\lambda_\pr{1-b}{j}	= \Ext(S,
										c_\pr{b}{d+(j-1)\nu+1})||\dots||\Ext(S,c_\pr{b}{d+j\nu})
									\end{displaymath}
							\end{itemize}
						\item\linelabel{ln:skstar1} At the conclusion of the $k\cdot\nu$ exchange
							rounds, each party updates his state to contain the
							secret key $\sk^*$ computed as follows:\linelabel{ln:skstar2} 
							\begin{displaymath}
								\sk^* = \SKE.\Gen\left(1^\kappa;\Lambda_{\sf
								out}\big((\lambda_\pr{b}{j})_{b\in\zo,j\in[k]}\big)\right)\;.
							\end{displaymath}
					\end{itemize}
			\end{enumerate}
		\item \textbf{Communication Phase} $\OurXi$ \\
			Each communication phase occurs over $r'\defeq\xi/v$ exchange-rounds. \\ \newcommand{\notbeta}{\bar{\beta}}
			Let $\beta\in\zo$ be the initiating party and let $\notbeta=1-\beta$.

					$P_\beta$ performs the following steps.
						\begin{itemize}
							\item\linelabel{ln:SKE-ciphs} Generate $c^*\gets\SKE.\Enc(\sk^*,\msg)$. 
							\item Break $c^*$ into $v$-bit blocks
							$c^*=c^*_1||\dots||c^*_{r'}$.
						\end{itemize}
					For $i'\in\{1,\dots,r'\}$:

						\begin{itemize}
							\item Let $i'' = r+i'$.
							\item $\OurXi_\pr{0}{i'}(\msg,
								m_\pr{0}{i''},\pk_{\notbeta},\state_\pr{\beta}{i''-1})$
								outputs $c_\pr{\beta}{i''}\gets\Sigma^{\E,
								S}(c^*_{i'},
								m_\pr{\beta}{i''}, \pk_{\notbeta})$.
							\item
								$\OurXi_\pr{1}{i'}(m_\pr{\notbeta}{i''},\pk_{\beta},\state_\pr{\notbeta}{i''-1})$
								outputs
								$c_\pr{\notbeta}{i''}\gets\E.\Enc(\pk_\beta,m_\pr{\notbeta}{i''})$.
							\item Both parties update their state to contain
								the transcript of all protocol messages
								exchanged so far.
						\end{itemize}

					After the $r'$ exchange-rounds,
						$P_{\notbeta}$ computes $c^{**}$ as follows:
						\begin{displaymath}
							c^{**} = \Ext(S, c_\pr{\beta}{r+1})||\dots||\Ext(S,
							c_\pr{\beta}{r+r'})\;.
						\end{displaymath}
						Then, $P_{\notbeta}$ outputs
						$\msg'\gets\SKE.\Dec(\sk^*,c^{**})$. (That is,
						$\OurXi_{1, \sf
						out}(\state_\pr{\notbeta}{r'})=\msg'$.)
		\end{linenumbers}\end{enumerate}
\end{definition}

\subsection{Proof that $(\OurPhi,\OurXi)$ Is a Subliminal Communication Scheme}

Finally, we give the proof of our main theorem. We recall the statement here.

\begin{customthm}{\ref{thm:subl}}
	% This is a copy of the theorem upstairs
	Assume there exists a pseudorandom key-exchange protocol.
	Then there is a multi-message subliminal communication scheme 
	$(\OurPhi,\OurXi)$, given in Definition~\ref{def:full-protocol}.
\end{customthm}
\begin{proof}
	We define three hybrids.
	\begin{itemize}
		\item {\sc Hybrid 0 (``real world''):}
			Parties execute $(\OurPhi,\OurXi)$.
		\item {\sc Hybrid 1:}
			Exactly like Hybrid 0, except that the seed $S$ in Phase~\ref{itm:seed} is replaced
			by a truly random $d$-bit string (the same string for both parties).
		\item {\sc Hybrid 2:}
			Exactly like Hybrid 1, except that the key exchange messages $\lambda_\pr{b}{j}$ in $\OurPhi$ are replaced by random strings.
			That is, Line~\lineref{ln:KE-msg} is replaced by:
			\begin{quote}
				$P_b$ 
				samples
				$\lambda_\pr{b}{j}
				\gets\zo^\ell$ at random.
			\end{quote}		
		\item {\sc Hybrid 3:}
			Exactly like Hybrid 2, except that the ciphertexts of $\SKE$ in $\OurXi$ are replaced by random strings.  That is, Line~\lineref{ln:SKE-ciphs} is replaced by:
			\begin{quote}
				Sample $c^*\gets\zo^\xi$ at random.
			\end{quote}
	\end{itemize}

	Hybrids 0 and 1 are indistinguishable by direct application of Theorem~\ref{thm:ext}.

	Hybrids 1 and 2 are indistinguishable by 
	the pseudorandomness of the key-exchange protocol (defined in Section~\ref{sec:prelims-crypto}).
	Note that it is essential that the $\Lambda$-transcript's indistinguishability from random holds 
	even in the presence of the established key $\sk^*$, since in our protocol $(\OurPhi,\OurXi)$,
	the later protocol messages are produced as a function of $\sk^*$.

	Hybrids 2 and 3 are indistinguishable because ciphertexts of $\SKE$ are pseudorandom in the absence of the corresponding secret key $\sk^*$. Because we already replaced the messages of $\Lambda$ with random messages independent of $\sk^*$, the distribution of all protocol messages in Hybrid 1 can be generated based on just the $\SKE$-ciphertexts $c^*$ (of line~\lineref{ln:SKE-ciphs}).	

	Finally, Hybrid 3 is indistinguishable from the ideal distribution 
	$${\sf Ideal}(\pk_1,\sk_1,\pk_2,\sk_2,\cM)$$ 
	from Definition~\ref{def:mm-subliminal} by Theorem~\ref{lem:comp}.
	Note that the rejection sampler $\Sigma^{\E,U_d}$ of Algorithm \ref{algo:rs}
	has a truly random $d$-bit seed, so to invoke Theorem~\ref{lem:comp} we rely on the fact that $S$ is truly random in Hybrid~3.
\end{proof}

\subsection{On the Setup Cost and Asymptotic Rate of $(\OurPhi,\OurXi)$}

\subh{Setup Cost.} The setup cost of our scheme can be broken down into
the costs of Step \ref{itm:seed} and Step \ref{itm:KE} as follows.

\begin{itemize}
	\item \emph{Step \ref{itm:seed}:} If our scheme is instantiated with the
		extractor $\Ext$ from Proposition~\ref{prop:auh}, then we need to
		establish a seed of length $\widetilde O(\log \kappa)$, which implies
		that $\widetilde O(\log\kappa)$ exchange-rounds are required in Step
		\ref{itm:seed}.  This is arguably the least efficient step in our
		scheme; this inefficiency stems from the use of the $\Gt$ extractor
		which only outputs one bit: to the best of our knowledge this the only
		extractor which applies to our setting. In Section~\ref{sec:improvs},
		we discuss ways in which this cost can be reduced under additional
		assumptions on the next-message distribution or the encryption scheme
		$\E$ by replacing $\Gt$ by extractors with longer outputs.
	\item \emph{Step \ref{itm:KE}:} The cost of this step is
		$k\cdot\frac{\ell}{v}$, where $k$ is the number of rounds of the
		key-exchange protocol $\Lambda$ that we use, $\ell$ is the length of
		messages in $\Lambda$ and $v$ is the output length of $\Ext$. 
		(Concretely, the Diffie-Hellman key exchange protocol achieves $k=1$
		and $\ell = O(\kappa)$.) If we use the extractor from
		Proposition~\ref{prop:auh} as $\Ext$, then we can achieve $v = c\log
		\kappa$ for any $c > 0$. This implies that the cost of Step \ref{itm:KE} is
		upper-bounded by $c'\frac{\kappa}{\log \kappa}$ for any $c'> 0$. Note
		that because the min-entropy of ciphertexts from $\E$ can be as small
		as $\omega(\log \kappa)$ and is \emph{a priori} unknown to the designer
		of the subliminal scheme, the output of $\Ext$ must be
		$O(\log \kappa)$.\footnote{Suppose not, \ie, suppose that the output of
		$\Ext$ were  $v\in\omega(\log\kappa)$ bits long. 
		Then the adversary could choose an encryption scheme whose ciphertexts 
		have min-entropy $z\in\omega(\log\kappa)\cap o(v)$ (\emph{e.g.},
		$z=\sqrt{v(\kappa)\log\kappa})$.
		Since the extractor output cannot have more min-entropy than its
		input, the extractor's output when evaluated on ciphertexts
		would be distinguishable from random $v$-bit strings.
		} The extractor $\Ext$ from Proposition~\ref{prop:auh} is
		optimal in this respect.
\end{itemize}

\subh{Asymptotic Rate.} The asymptotic rate of our scheme depends on the
output length of $\Ext$. Using the extractor from Proposition~\ref{prop:auh},
we can subliminally embed $c\log \kappa$ random bits per ciphertext of $\E$ and
hence achieve an asymptotic rate of $c\log\kappa$ for any $c>0$. As in the
discussion regarding the cost of Step \ref{itm:KE}, this is optimal given that
the min-entropy of ciphertexts can be as small as $\omega(\log \kappa)$.

\subh{Trade-off Between Running Time and Rate.} Note that the parameter
$c$ from the previous paragraph controls the trade-off between the running time
of our scheme and its asymptotic rate. Indeed, the expected running time of the
rejection sampler defined in Algorithm~\ref{algo:rs} is $O(\kappa^c)$, when 
embedding $c\log \kappa$ random bits. This trade-off is inherent
to rejection sampling, and it is an interesting open question to determine
whether it can be improved by an alternative technique.

\section{Improving Setup Cost}
\label{sec:improvs}

In this section, we present alternative constructions which improve the setup
cost of subliminal communication under additional assumptions, either on the
next-message distribution $\cM$, or on the public-key encryption scheme $\E$.
These additional assumptions allow us to replace the ``greater-than'' extractor
in Step \ref{itm:seed} of Definition \ref{def:outline} with extractors with
longer output, thus reducing the number of exchange-rounds required to
establish the shared seed.  Section~\ref{sec:improv-succinct} gives
a construction when $\E$ is ``succinct'' (\emph{i.e.}, has a constant expansion
factor).  Section~\ref{sec:improv-minent} presents a construction when $\cM$
has a known amount of min-entropy. Both these constructions yield a seed
establishment in two exchange-rounds.

\subsection{\emph{Succinct} Encryption Schemes}
\label{sec:improv-succinct}

Let us suppose that the adversarially chose scheme $\E$ is \emph{succinct}.
Here we define \emph{succinct} as having an expansion factor less than $2$.
Recall that the expansion factor is defined to be the ratio of ciphertext
length to plaintext length. Under this assumption, we can improve the number of
rounds in the construction of Section~\ref{sec:establish-seed} by replacing the
extractor $\Gt$ by the extractor $\textsc{Ble}$ from \cite{improv}. Note that
this extractor requires the min-entropy rate of the sources to be slightly
above $\frac{1}{2}$, yet ciphertexts of $\E$ only have min-entropy $\omega(\log
\kappa)$.  However, the succinctness assumption combined with semantic security
implies that cipertexts have sufficiently large HILL entropy: they are
computationally indistinguishable from sources of min-entropy rate slightly
above $\frac{1}{2}$. Since the extractor is a polynomial time algorithm, its
output when computed on ciphertexts from $\E$ will be computationally
indistinguishable from the uniform distribution.  Formally, we prove the
following theorem.

\begin{theorem}
	\label{thm:cc}
	Let us denote by $n$ (resp. $p$) the bit-lengths of ciphertexts (resp.
	plaintexts) from $\E$. 
	Let $\textsc{Ble}$ be the function constructed in \cite{improv} and let us
	denote by $\sB:\zo^n\times\zo^n\to\zo^v$ its truncation to
	$v=O(\log\kappa)$ bits.  With the same notations as Theorem~\ref{thm:ext},
	define:
	\begin{displaymath}
		S = \big(\sB(C_\pr{0}{1},C_\pr{0}{2}),\sB(C_\pr{1}{1},C_\pr{1}{2}),\dots,
		\sB(C_\pr{0}{2\rho-1},C_\pr{0}{2\rho}),\sB(C_\pr{1}{2\rho-1},C_\pr{1}{2\rho})\big).
	\end{displaymath}
	There is a negligible function $\eps$ such that
	$
		\sdist{(K_0,K_1, S)
		- (K_0, K_1, U_{2\rho v})}\leq~\eps
	$.
\end{theorem}

\begin{proof}
	For any $\gamma > 0$, the function $\textsc{Ble}$ from \cite{improv}  can
	extract $2p-n-\gamma p$ random bits with bias $\eps=2^{-1/2\gamma p}$ from
	two independent $n$-bit sources of min-entropy $p$. Note that for $n\leq
	(1-\gamma)2p$, $2p-n-\gamma p\geq \gamma p$, and by semantic security of
	$\E$, we have that $n$ (and hence $p$) are super-logarithmic. Hence the
	function $\sB$ defined in Theorem~\ref{thm:cc} is well-defined and has bias
	$\eps\in\negl$ for independent sources of min-entropy $p$.

	For public key $\pk$, define $Y^{\pk}=\E.\Enc(\pk, U_p)$ the random
	variable obtained by encrypting a message chosen uniformly at random. By
	correctness of $\E$, $Y^{\pk}$ is a flat source supported on $\zo^p$, hence
	$\Me(Y^{\pk}) = p$. Hence for $Z^{\pk}$ an independent copy of $Y^{\pk}$, we have
	$
		\sdist{\sB(Y^{\pk}, Z^{\pk})-U_v}\leq \eps(\kappa)
	$.
	By semantic indistinguishability of $\E$, we have that for any ciphertext
	$C$ of $\E$, $C\compIndist Y^\pk$; \emph{i.e.}, $C$ has HILL-entropy $p$.
	Adapting the proof of Lemma~\ref{lemma:2ext-comp} and using the same
	notation, since $v=O(\log\kappa)$, we get
	$
		\sdist{\big(PK, SK, \sB(C^{PK}_0, C^{PK}_1)\big)
		-
		\big(PK, SK, U_v\big)}\leq\eps
	$.
	From there, the proof of Theorem~\ref{thm:cc} follows verbatim the
	proof of Theorem~\ref{thm:ext}.
\end{proof}

\begin{remark}
	Theorem~\ref{thm:cc} implies in particular that a $\widetilde
	O(\log\kappa)$ random seed can be established in step \ref{itm:seed} of our
	subliminal scheme in only two exchange-rounds of communication, whereas the
	construction of Theorem~\ref{thm:ext} with the $\Gt$ extractor requires
	$\widetilde O(\log \kappa)$ exchange-rounds. 
\end{remark}

\subsection{Next-Message Distributions with Min-Entropy}
\label{sec:improv-minent}

The previous subsection reduced the number of rounds to establish the shared
seed by using a two-source extractor with more than one bit of output (unlike
the $\Gt$ extractor). However, this required $\E$ to be succinct to
guarantee that ciphertexts of $\E$ have sufficient HILL-entropy. In this
section, we observe that if the message distribution itself has enough
min-entropy, then we can use the \emph{plaintext} and corresponding
\emph{ciphertext} as a pair of sources to extract from, and obtain a similar
improvement without requiring that $\E$ be succinct. Note that this
construction again exploits the semantic security of $\E$ to guarantee that
the plaintext and the ciphertext are indistinguishable from independent sources
to the two-source extractor.

All we require of $\E$ in this subsection is that the ciphertext distribution
has min-entropy at least $\omega(\log\kappa)$, which follows from semantic
security (Lemma~\ref{lemma:cond-min-ent}); however, we additionally require
that the next message distribution $\mathcal{M}$ have min-entropy rate above
$\frac{1}{2}$.  While this is quite a lot of min-entropy to demand from
$\mathcal{M}$, we note that the precise requirement of $\frac{1}{2}$ arises
from the current state of the art in two-source extractors, and improved
two-source extractor constructions would directly imply improvements in the
min-entropy requirement of our construction. Recent research in two-source
extraction has been quite productive; with luck, future advances will provide
us with an alternative which demands much less entropy from $\mathcal{M}$.

\begin{theorem}
	\label{thm:mc}
	Suppose that $k_1$ is such that for all $k_2=\omega(\log\kappa)$ there
	exists a negligible function $\eps$ and a $(k_1,k_2,\eps)$-two source
	extractor $\TwoExt$ with output length $\ell=O(\log\kappa)$.  
	and let $\rho\in\poly$. Define random variables as follows.
	\begin{itemize}
		\item For $b\in\zo$, let $K_b=(PK_b, SK_b)=\E.\Gen(1^\kappa)$.
		\item For $b\in\zo$ and $i\in[\rho]$, let $M_\pr{b}{i}$ and
			$C_\pr{b}{i} = \E.\Enc(PK_{1-b}, m_\pr{b}{i})$ be the messages and
			ciphertexts exchanged between $P_0$ and $P_1$ in $\rho$
			exchange-rounds.
		\ifsubmission
		\item $
		S = \big(
		\TwoExt(M_\pr{0}{1},C_\pr{0}{1}), \TwoExt(M_\pr{1}{1}, C_\pr{1}{1}),
		\dots,
		\TwoExt(M_\pr{1}{\rho}, C_\pr{1}{\rho})
		\big)
			$.
		\else
		\item $
		S = \big(
		\TwoExt(M_\pr{0}{1},C_\pr{0}{1}), \TwoExt(M_\pr{1}{1}, C_\pr{1}{1}),
		\dots,
		\TwoExt(M_\pr{0}{\rho},C_\pr{0}{\rho}), \TwoExt(M_\pr{1}{\rho}, C_\pr{1}{\rho})
		\big)
			$.
		\fi
	\end{itemize}
	Then if the next message distribution $\mathcal{M}({\sf conv})$ has min-entropy at
	least $k_1$:
	\begin{displaymath}
		\sdist{(K_0,K_1, S)
		- (K_0, K_1, U_{\rho})}\leq \eps\;.
	\end{displaymath}
\end{theorem}

\begin{proof}
	Let us consider $M$ and $M'$ two independent samples from the next-message
	distribution and let us denote by $C^{PK}_M$ and $C^{PK}_{M'}$ their
	encryption under key $PK$.
	Consider a two-source extractor $\TwoExt$ as in the statement of the
	theorem. By semantic security, it follows that for some negligible function
	$\eps'$:
	\begin{displaymath}
		\TwoExt(M, C_M^{PK})\compIndist_{\eps'}
		\TwoExt(M', C_M^{PK})\;.
	\end{displaymath}
	By property of $\TwoExt$ it follows that $\sdist{\TwoExt(M',
	C_M^{PK})-U_\ell}\leq\eps$. Since $\ell = O(\log \kappa)$, we can prove
	similarly to Lemma~\ref{lemma:2ext-comp}:
	\begin{displaymath}
		\sdist{\big(PK, SK, \TwoExt(M,C^{PK}_M)\big)-(PK, SK,
		U_{\ell})}\leq\eps+\eps'
		\;.
	\end{displaymath}
We can now conclude similarly to the proof of Theorem~\ref{thm:ext}.
\end{proof}

The following result from \cite{Raz05} implies that for any $\delta > 0$, when
$k_1=\frac{1}{2}+\delta$, two-source extractors exist which may be used in the
above theorem. The output length of such extractors is logarithmic, implying
that a $\widetilde O(\log \kappa)$-bit random seed $S$ can be established in
two exchange-rounds of communication.

\begin{lemma}[Theorem 1 in \cite{Raz05}]
	For any $\delta>0$, $k_1\geq\frac{1}{2}+\delta$ and
	$k_2=\omega(\log\kappa)$, there exists a negligible function $\eps$ and
	a $(k_1,k_2,\eps)$-two source extractor with output length
	$\ell=O(\log\kappa)$.
\end{lemma}

\section{Open problems}
\label{sec:open}

\subh{Deterministic Extraction.}
Our impossibility result in Theorem~\ref{thm:impossibility-1} holds because the
adversary can choose the encryption scheme $\E$ as a function of a given
candidate subliminal scheme. However, note that under the additional assumption
that $\E$ is restricted to a predefined class $C$ of encryption schemes, we
could bypass this impossibility as long as a deterministic extractor that can
extract randomness from ciphertexts of any encryption scheme in $C$ exists. We
are only aware of two deterministic extractors leading to a positive result for
restricted classes of encryption schemes:
\begin{itemize}
	\item if an upper bound on the circuit size of $\E$ is known, then we can
	use the deterministic extractor from \cite{TV00}. This extractor relies on
	strong complexity-theoretic assumptions and requires the sources to have
	min-entropy $(1-\gamma)n$ for some unspecified constant $\gamma$.
	\item if $\E$ is computed by a circuit of constant depth (\textbf{AC$^0$}),
	then the deterministic extractor of \cite{Viola11} can be used and requires
	$\sqrt{n}$ min-entropy.
\end{itemize}

Note that both these extractors have a min-entropy requirement which is too
strict to be directly applicable to ciphertexts of arbitrary encryptions
schemes, but it would be interesting to give improved constructions for the
specific case of encryption schemes. This would also have direct implications
for the efficiency of the subliminal scheme we construct in
Section~\ref{sec:overt}: indeed, one could then skip Step \ref{itm:seed} and
use a deterministic extractor directly in Steps \ref{itm:KE} and and
\ref{itm:comm}, thus saving $\widetilde O(\log \kappa)$ exchange-rounds in the
setup phase.

\subh{Multi-Source Extraction.}
Another interesting question is whether multi-source extractors for the
specific case when the sources are independent and identically distributed can
achieve better parameters than extractors for general independent sources. We
already saw that a very simple extractor (namely, the ``greater-than''
function) works for i.i.d. sources and extracts one bit with negligible
bias, even when the sources only have $\omega(\log\kappa)$ min-entropy. The
non-constructive result of \cite{CG88} guarantees the existence of a two-source
extractor of negligible bias and output length $\omega(\log\kappa)$ for sources
of min-entropy $\omega(\log\kappa)$. However, known \emph{explicit}
constructions are far from achieving the same parameters, and improving them in
the specific case of identically distributed sources is an interesting open
problem which was also mentioned in \cite{BIW04}.

\ifsubmission\else
  \section*{Acknowledgements}

  We are grateful to Omer Paneth and Adam Sealfon for insightful remarks on an earlier draft of this paper.
\fi

\bibliographystyle{alpha}
\bibliography{backdoorrefs,bib,refs-sp}

\newcommand{\etalchar}[1]{$^{#1}$}
\begin{thebibliography}{ACM{\etalchar{+}}14}

\bibitem[ACM{\etalchar{+}}14]{ACMPS14}
Per Austrin, Kai{-}Min Chung, Mohammad Mahmoody, Rafael Pass, and Karn Seth.
\newblock On the impossibility of cryptography with tamperable randomness.
\newblock In Garay and Gennaro \cite{DBLP:conf/crypto/2014-1}, pages 462--479.

\bibitem[Auc98]{DBLP:conf/ih/1998}
David Aucsmith, editor.
\newblock {\em Information Hiding, Second International Workshop, Portland,
  Oregon, USA, April 14-17, 1998, Proceedings}, volume 1525 of {\em Lecture
  Notes in Computer Science}. Springer, 1998.

\bibitem[BC05]{BC05}
Michael Backes and Christian Cachin.
\newblock Public-key steganography with active attacks.
\newblock In Joe Kilian, editor, {\em Theory of Cryptography, Second Theory of
  Cryptography Conference, {TCC} 2005, Cambridge, MA, USA, February 10-12,
  2005, Proceedings}, volume 3378 of {\em Lecture Notes in Computer Science},
  pages 210--226. Springer, 2005.

\bibitem[BIW04]{BIW04}
Boaz Barak, Russell Impagliazzo, and Avi Wigderson.
\newblock Extracting randomness using few independent sources.
\newblock In {\em 45th Symposium on Foundations of Computer Science {(FOCS}
  2004), 17-19 October 2004, Rome, Italy, Proceedings}, pages 384--393, 2004.

\bibitem[BJK15]{BJK15}
Mihir Bellare, Joseph Jaeger, and Daniel Kane.
\newblock Mass-surveillance without the state: Strongly undetectable
  algorithm-substitution attacks.
\newblock In Indrajit Ray, Ninghui Li, and Christopher Kruegel, editors, {\em
  Proceedings of the 22nd {ACM} {SIGSAC} Conference on Computer and
  Communications Security, Denver, CO, USA, October 12-6, 2015}, pages
  1431--1440. {ACM}, 2015.

\bibitem[BL17]{BL17}
Sebastian Berndt and Maciej Li{\'{s}}kiewicz.
\newblock Algorithm substitution attacks from a steganographic perspective.
\newblock In Bhavani~M. Thuraisingham, David Evans, Tal Malkin, and Dongyan Xu,
  editors, {\em Proceedings of the 2017 {ACM} {SIGSAC} Conference on Computer
  and Communications Security, {CCS} 2017, Dallas, TX, USA, October 30 -
  November 03, 2017}, pages 1649--1660. {ACM}, 2017.

\bibitem[BPR14]{BPR14}
Mihir Bellare, Kenneth~G. Paterson, and Phillip Rogaway.
\newblock Security of symmetric encryption against mass surveillance.
\newblock In Garay and Gennaro \cite{DBLP:conf/crypto/2014-1}, pages 1--19.

\bibitem[Cac98]{Cac98}
Christian Cachin.
\newblock An information-theoretic model for steganography.
\newblock In Aucsmith \cite{DBLP:conf/ih/1998}, pages 306--318.

\bibitem[CG88]{CG88}
Benny Chor and Oded Goldreich.
\newblock Unbiased bits from sources of weak randomness and probabilistic
  communication complexity.
\newblock {\em {SIAM} J. Comput.}, 17(2):230--261, 1988.

\bibitem[CK16]{CK16}
Aloni Cohen and Saleet Klein.
\newblock The {GGM} function family is a weakly one-way family of functions.
\newblock In Martin Hirt and Adam~D. Smith, editors, {\em Theory of
  Cryptography - 14th International Conference, {TCC} 2016-B, Beijing, China,
  October 31 - November 3, 2016, Proceedings, Part {I}}, volume 9985 of {\em
  Lecture Notes in Computer Science}, pages 84--107, 2016.

\bibitem[CZ16]{CZ16}
Eshan Chattopadhyay and David Zuckerman.
\newblock Explicit two-source extractors and resilient functions.
\newblock In {\em Proceedings of the 48th Annual {ACM} {SIGACT} Symposium on
  Theory of Computing, {STOC} 2016, Cambridge, MA, USA, June 18-21, 2016},
  pages 670--683, 2016.

\bibitem[DEOR04]{improv}
Yevgeniy Dodis, Ariel Elbaz, Roberto Oliveira, and Ran Raz.
\newblock Improved randomness extraction from two independent sources.
\newblock In {\em 7th International Workshop on Approximation Algorithms for
  Combinatorial Optimization Problems, {APPROX} 2004, and 8th International
  Workshop on Randomization and Computation, {RANDOM} 2004, Cambridge, MA, USA,
  August 22-24, 2004, Proceedings}, pages 334--344, 2004.

\bibitem[DGG{\etalchar{+}}15]{DGGJR15}
Yevgeniy Dodis, Chaya Ganesh, Alexander Golovnev, Ari Juels, and Thomas
  Ristenpart.
\newblock A formal treatment of backdoored pseudorandom generators.
\newblock In Elisabeth Oswald and Marc Fischlin, editors, {\em Advances in
  Cryptology - {EUROCRYPT} 2015 - 34th Annual International Conference on the
  Theory and Applications of Cryptographic Techniques, Sofia, Bulgaria, April
  26-30, 2015, Proceedings, Part {I}}, volume 9056 of {\em Lecture Notes in
  Computer Science}, pages 101--126. Springer, 2015.

\bibitem[DH76]{DH76}
Whitfield Diffie and Martin~E. Hellman.
\newblock New directions in cryptography.
\newblock {\em {IEEE} Trans. Information Theory}, 22(6):644--654, 1976.

\bibitem[GG14]{DBLP:conf/crypto/2014-1}
Juan~A. Garay and Rosario Gennaro, editors.
\newblock {\em Advances in Cryptology - {CRYPTO} 2014 - 34th Annual Cryptology
  Conference, Santa Barbara, CA, USA, August 17-21, 2014, Proceedings, Part
  {I}}, volume 8616 of {\em Lecture Notes in Computer Science}. Springer, 2014.

\bibitem[GGM86]{GGM86}
Oded Goldreich, Shafi Goldwasser, and Silvio Micali.
\newblock How to construct random functions.
\newblock {\em J. {ACM}}, 33(4):792--807, 1986.

\bibitem[Gol11]{Gol02}
Oded Goldreich.
\newblock Studies in complexity and cryptography.
\newblock chapter The GGM Construction Does NOT Yield Correlation Intractable
  Function Ensembles, pages 98--108. Springer-Verlag, Berlin, Heidelberg, 2011.

\bibitem[HLv02]{HLA02}
Nicholas~J. Hopper, John Langford, and Luis {von Ahn}.
\newblock Provably secure steganography.
\newblock In Moti Yung, editor, {\em Advances in Cryptology - {CRYPTO} 2002,
  22nd Annual International Cryptology Conference, Santa Barbara, California,
  USA, August 18-22, 2002, Proceedings}, volume 2442 of {\em Lecture Notes in
  Computer Science}, pages 77--92. Springer, 2002.

\bibitem[Mic16]{Micali16}
Silvio Micali.
\newblock {ALGORAND:} the efficient and democratic ledger.
\newblock {\em CoRR}, abs/1607.01341, 2016.

\bibitem[Mit99]{Mit99}
Thomas Mittelholzer.
\newblock An information-theoretic approach to steganography and watermarking.
\newblock In Andreas Pfitzmann, editor, {\em Information Hiding, Third
  International Workshop, IH'99, Dresden, Germany, September 29 - October 1,
  1999, Proceedings}, volume 1768 of {\em Lecture Notes in Computer Science},
  pages 1--16. Springer, 1999.

\bibitem[Raz05]{Raz05}
Ran Raz.
\newblock Extractors with weak random seeds.
\newblock In Harold~N. Gabow and Ronald Fagin, editors, {\em Proceedings of the
  37th Annual {ACM} Symposium on Theory of Computing, Baltimore, MD, USA, May
  22-24, 2005}, pages 11--20. {ACM}, 2005.

\bibitem[Sim83]{Sim83}
Gustavus~J. Simmons.
\newblock The prisoners' problem and the subliminal channel.
\newblock In David Chaum, editor, {\em Advances in Cryptology, Proceedings of
  {CRYPTO} '83, Santa Barbara, California, USA, August 21-24, 1983.}, pages
  51--67. Plenum Press, New York, 1983.

\bibitem[SZ94]{SZ94}
Aravind Srinivasan and David Zuckerman.
\newblock Computing with very weak random sources.
\newblock In {\em 35th Annual Symposium on Foundations of Computer Science,
  Santa Fe, New Mexico, USA, 20-22 November 1994}, pages 264--275. {IEEE}
  Computer Society, 1994.

\bibitem[TV00]{TV00}
Luca Trevisan and Salil~P. Vadhan.
\newblock Extracting randomness from samplable distributions.
\newblock In {\em 41st Annual Symposium on Foundations of Computer Science,
  {FOCS} 2000, 12-14 November 2000, Redondo Beach, California, {USA}}, pages
  32--42, 2000.

\bibitem[Vad12]{vad10}
Salil~P. Vadhan.
\newblock Pseudorandomness.
\newblock {\em Foundations and Trends® in Theoretical Computer Science},
  7(1–3):1--336, 2012.

\bibitem[vAH04]{AH04}
Luis von Ahn and Nicholas~J. Hopper.
\newblock Public-key steganography.
\newblock In Christian Cachin and Jan Camenisch, editors, {\em Advances in
  Cryptology - {EUROCRYPT} 2004, International Conference on the Theory and
  Applications of Cryptographic Techniques, Interlaken, Switzerland, May 2-6,
  2004, Proceedings}, volume 3027 of {\em Lecture Notes in Computer Science},
  pages 323--341. Springer, 2004.

\bibitem[Vio11]{Viola11}
Emanuele Viola.
\newblock Extractors for circuit sources.
\newblock In Rafail Ostrovsky, editor, {\em {IEEE} 52nd Annual Symposium on
  Foundations of Computer Science, {FOCS} 2011, Palm Springs, CA, USA, October
  22-25, 2011}, pages 220--229. {IEEE} Computer Society, 2011.

\bibitem[YY96a]{YY96a}
Adam~L. Young and Moti Yung.
\newblock Cryptovirology: Extortion-based security threats and countermeasures.
\newblock In {\em 1996 {IEEE} Symposium on Security and Privacy, May 6-8, 1996,
  Oakland, CA, {USA}}, pages 129--140. {IEEE} Computer Society, 1996.

\bibitem[YY96b]{YY96b}
Adam~L. Young and Moti Yung.
\newblock The dark side of "black-box" cryptography, or: Should we trust
  capstone?
\newblock In Neal Koblitz, editor, {\em Advances in Cryptology - {CRYPTO} '96,
  16th Annual International Cryptology Conference, Santa Barbara, California,
  USA, August 18-22, 1996, Proceedings}, volume 1109 of {\em Lecture Notes in
  Computer Science}, pages 89--103. Springer, 1996.

\bibitem[YY97]{YY97}
Adam~L. Young and Moti Yung.
\newblock Kleptography: Using cryptography against cryptography.
\newblock In Walter Fumy, editor, {\em Advances in Cryptology - {EUROCRYPT}
  '97, International Conference on the Theory and Application of Cryptographic
  Techniques, Konstanz, Germany, May 11-15, 1997, Proceeding}, volume 1233 of
  {\em Lecture Notes in Computer Science}, pages 62--74. Springer, 1997.

\bibitem[ZFK{\etalchar{+}}98]{ZFKPPWWW98}
Jan Z{\"{o}}llner, Hannes Federrath, Herbert Klimant, Andreas Pfitzmann, Rudi
  Piotraschke, Andreas Westfeld, Guntram Wicke, and Gritta Wolf.
\newblock Modeling the security of steganographic systems.
\newblock In Aucsmith \cite{DBLP:conf/ih/1998}, pages 344--354.

\end{thebibliography}

\newpage
\appendix
\section{Steganographic communication schemes}
\label{a:imposs}

Here, we give a formal definition of a steganographic communication scheme,
formalized as a protocol consisting of $r$ exchange-rounds
in which two parties $P_0$ and $P_1$
engage in communication distributed according to a cover distribution $\cC$,
and in which $P_0$ steganographically transmits a hidden message $\sf msg$ to $P_1$.
The cover distribution $\cC$ is a parameter of the communication scheme.

\begin{definition}[Steganographic communication scheme]
\label{def:stego}
A \emph{steganographic communication scheme} with respect to a cover distribution $\mathcal{C}$ is a two-party protocol:
\begin{displaymath}
\Pi^\cC=\bigl(\Pi_\pr{0}{1}^\cC,\Pi^\cC_\pr{1}{1},\Pi^\cC_\pr{0}{2},\Pi^\cC_\pr{1}{2},
	\dots,\Pi^\cC_\pr{0}{r},\Pi^\cC_\pr{1}{r},;\Pi^\cC_{\party{1},{\sf out}}\bigr)
\end{displaymath}
	where $r\in\poly$ and where each $\Pi^\cC_\pr{b}{i}$ is a PPT
	algorithm with oracle access to a $\cC-$sampler.
	Party $P_0$ is assumed to receive as input a message $\msg$ 
	that is to be conveyed to $P_1$ in an undetectable fashion.  
	The protocol must satisfy the following syntax,
	correctness and security guarantees.

\begin{itemize}
\item \textbf{Syntax.} For each $i=1,\dots, r$:
	\begin{enumerate}
		\item $P_0$ draws $(c_\pr{0}{i},\state_\pr{0}{i})\leftarrow\Pi_\pr{0}{i}^\cC(\msg,\state_\pr{0}{i-1})$, locally stores $\state_\pr{0}{i}$, and sends
			$c_\pr{0}{i}$ to $P_1$.
		\item $P_1$ draws $(c_\pr{1}{i},\state_\pr{1}{i})\leftarrow\Pi_\pr{1}{i}^\cC(\state_\pr{1}{i-1})$, locally stores $\state_\pr{1}{i}$,
			and sends $c_\pr{1}{i}$ to $P_0$. 
	\end{enumerate}
		After the $r$th exchange-round, $P_1$ computes ${\sf msg}'=\Pi^\cC_{\party{1},{\sf
		out}}(c_\pr{0}{1},c_\pr{1}{1},\dots,c_\pr{0}{r},c_\pr{1}{r})$ and halts.
\item\textbf{Correctness.} For any $\msg\in\{0,1\}^\kappa$, $\msg'=\msg$ with overwhelming probability.
		The probability is taken over the randomness of all of the $\Pi^\cC_\pr{b}{i}$ and their $\cC$-samples.

\item\textbf{Steganographic Indistinguishability w.r.t. $\cM$.}
	The following distribution
	$$\Big\{(c_\pr{0}{1},c_\pr{1}{1},\dots,c_\pr{0}{r},c_\pr{1}{r}):\msg\gets\cM, (c_\pr{0}{i},\state_\pr{0}{i})\gets\Pi_\pr{0}{i}^\cC(\msg,\state_\pr{0}{i-1}), (c_\pr{1}{i},\state_\pr{1}{i})\gets\Pi_\pr{1}{i}^\cC(\state_\pr{1}{i-1})\Bigr\}$$
	is computationally indistinguishable from the cover distribution of length $2r-1$. (Note that $\cC_{2r-1}$ is not necessarily a product distribution.) 

\end{itemize}
\end{definition}

%Observe that the final condition in the above definition, \emph{steganographic indistinguishability w.r.t. $\cM$}, is parametrized by a message distribution $\cM$.
Depending on the specific application at hand, a steganographic communication scheme might require that steganographic indistinguishability hold w.r.t. \emph{all} message distributions $\cM$,
or alternatively, only require that it hold w.r.t. certain specific message
distributions (\emph{e.g.}, uniformly random messages).

Definition~\ref{def:stego} concerns the transmission of only a single message. There is a natural generalization of this definition to a more general \emph{multi-message} version, analogous to how Definition~\ref{def:mm-subliminal} is a multi-message generalization of Definition~\ref{def:subliminal} in the context of subliminal communication schemes. 
Much as observed in Remark~\ref{rmk:equivalence-mm} in the context of subliminal communication schemes, the natural multi-message generalization of Definition~\ref{def:stego} is \emph{equivalent} to Definition~\ref{def:stego} in the sense that the existence of any single-message scheme easily implies a multi-message scheme and vice versa; the multi-message version of the definition is mainly useful when considering the asymptotic efficiency of schemes. We present only Definition \ref{def:stego} here as the simpler single-message definition suffices in the context of an impossibility result.

The existing notions of public-key steganography
and steganographic key exchange, as well as the \emph{subliminal communication schemes} introduced in this work, are all instantiations of the more general definition of a (multi-message) steganographic communication scheme.% 
\footnote{Public-key steganography protocols can be thought of as being parametrized by the public/private key pairs of the communicating parties, which are initially sampled by the parties according to some key generation algorithm. However, the definition of a steganographic communication protocol as stated is not parametrized in an analogous way. This syntactic discrepancy can be resolved by equivalently thinking of the public-key steganography as a \emph{family} of steganographic communication protocols each of which has the key pairs hardwired, induced by the sampling of key pairs according to the key generation algorithm; and then requiring the steganographic indistinguishability guarantee to hold only with overwhelming probability over key generation.}

\subsection{Proof of Theorem~\ref{thm:impossibility-2}}

In this subsection we give the proof of impossibility of non-trivial steganographic communication in the presence of an adversarially chosen cover distribution. We recall the statement of Theorem~\ref{thm:impossibility-2} from Section \ref{sec:imposs-stego}.

\begin{customthm}{\ref{thm:impossibility-2}}
	Let $\Pi$ be a steganographic communication scheme. Then for any $k\in\NN$,
		there exists a cover distribution $\cC$ of conditional min-entropy $k$ such
		that the steganographic indistinguishability of $\Pi$ does not hold for
		more than one message.
\end{customthm}
\begin{proof}
	The proof exploits the impossibility of extracting more than one bit
	from Santha-Vazirani (SV) sources (cf. Proposition 6.6 in \cite{vad10}).
	Recall that a $\delta$-SV source $X\in\zo^n$ is defined by the following
	conditions:
	\begin{displaymath}
		\delta\leq \Pr[X_i=1\,|\,X_1=x_1,\dots,X_{i-1}=x_{i-1}]\leq 1-\delta
	\end{displaymath}
	for all $i\in[n]$ and all $(x_1,\dots,x_{i-1})\in\zo^{i-1}$.

	Consider $\delta>0$ and $\Pi$ a steganographic communication scheme with
	$r$-rounds.  Let $n\in\NN$ be such that $\delta$-SV sources of length $n$
	have conditional min-entropy $k$ when interpreted as $2r-1$ blocks of size
	$\frac{n}{2r-1}$. Assume for contradiction that $\Pi$ can embed at least
	two messages $m_1$ and $m_2$. Steganographic indistinguishability applied
	to the case where $\msg\gets\{m_1, m_2\}$ implies the existence of
	a negligible $\eps$ such that $\bo[\Pi_{2,\sf out}(\cC_{2r-1})=m_1]$ is
	a random bit with bias at most $\eps$.

	But Proposition 6.6 in \cite{vad10} implies the existence of
	a $\delta$-SV source $\cC_{2r-1}$ such that $\Pr[\Pi_{2,\sf
	out}(\cC_{2r-1}) = m_1] \geq 1-\delta$ or $\Pr[\Pi_{2,\sf
	out}(\cC_{2r-1}) = m_1] \leq \delta$. In other words, $\bo[\Pi_{2, \sf
	out}(\cC_{2r-1}=m_1]$ has bias at least $\frac{1}{2}-\delta$. For
	$\delta<\frac{1}{2}-\eps$, this is a contradiction.
\end{proof}

\begin{remark}
	At first glance, it may seem that one could bypass this impossibility and communicate more than one bit of information by sequentially repeating the same steganographic scheme.  However, one could then apply Theorem~\ref{thm:impossibility-2} to the sequential composition and obtain a new cover distribution which would break this scheme.
\end{remark}

\section{Min-Entropy of Ciphertexts}
\label{a:me}

\begin{customlemma}{\ref{lemma:cond-min-ent}}
	Let $\PKE$ be a semantically secure encryption scheme. Then there exists
	a negligible function $\eps$ such that for all $\kappa\in\NN$,
	$m\in\Msg_{\kappa}$, writing $C^{pk}_m\sim\PKE.\Enc(pk,m)$:
	\begin{displaymath}
		\Pr\left[(pk, sk)\gets\PKE.\Gen(1^\kappa):
		\Me\big(C_m^{pk}\big)\geq\log\frac{1}{\eps(\kappa)}\right]
		\geq 1-\eps(\kappa)\;.
	\end{displaymath}
\end{customlemma}

\begin{proof}
	Let us assume by contradiction that there exists $p\in\poly(\kappa)$,
	a countably infinite set $I$ and a family $\{m_\kappa\}_{\kappa\in I}$ such
	that for all $\kappa\in I$:
	\begin{equation}
		\label{eq:cond-ent}
		\Pr\left[(pk, sk)\gets\PKE.\Gen(1^\kappa):
		\Me\big(C_{m_\kappa}^{pk}\big)\leq\log p(\kappa)\right]
		\geq \frac{1}{p(\kappa)}\;.
	\end{equation}
	Since $\PKE$ is non-trivial, for all $\kappa\in I$ there exists
	$m_{\kappa}'\neq m_{\kappa}$.

	Consider the distinguisher $\sD$ which on input $(pk, c)$ runs as follows:
	\begin{enumerate}
		\item Sample encryption $e\gets\PKE.\Enc(pk, m_{\kappa})$.
		\item If $e=c$ output 1, otherwise output 0.
	\end{enumerate}

	Writing $(PK, SK)\sim\PKE.\Gen(1^\kappa)$, we claim that $\sD$
	distinguishes $\big(PK, \Enc(PK, m_\kappa)\big)$ from $\big(PK, \Enc(PK,
	m'_\kappa)\big)$ with non-negligible advantage.  First note that by
	\eqref{eq:cond-ent}, with probablity at least $\frac{1}{p(\kappa)}$ over
	the draw of $pk$, $\Enc(pk, m_\kappa)$ has collision probability at least
	$\frac{1}{p(\kappa)^2}$. By definition, $\sD$ outputs 1 on input $(pk,
	\Enc(pk, m_{\kappa}))$ if and only if a collision occurs in step 1. Hence:
	\begin{displaymath}
		\Pr[\sD(PK, \Enc(PK, m_\kappa))=1]\geq \frac{1}{p(\kappa)^3}\;.
	\end{displaymath}

	Second, since $\PKE$ is correct, there exists a negligible function $\eps'$
	such that $\Enc(pk, m_\kappa)=\Enc(pk,m_\kappa')$ with probability at most
	$\eps'$. Hence:
	\begin{displaymath}
		\Pr[\sD(PK, \Enc(PK, m'_\kappa))=1]\leq \eps'(\kappa)\;.
	\end{displaymath}
	The previous two inequalities together imply that for large enough
	$\kappa\in I$:
	\begin{displaymath}
		\big|\Pr[\sD(PK, \Enc(PK, m_\kappa))=1]
		-
		\Pr[\sD(PK, \Enc(PK, m'_\kappa))=1]\big|
		\geq \frac{1}{2p(\kappa)^3}\;,
	\end{displaymath}
	which contradicts the semantic security of $\PKE$.
\end{proof}

\section{Greater-Than is a Same-Source Extractor}
\label{a:gt}

\begin{customlemma}{\ref{lemma:gt}}
For any $k\leq n$, $\Gt$ is a $(k, \frac{1}{2^k})$ same-source extractor.
\end{customlemma}

\begin{proof}
	Let $X$ and $Y$ be two i.i.d. variables with $\Me(X)=\Me(Y)\geq k$. Then
	\begin{displaymath}
		2\Pr[X\geq Y] = \Pr[X\geq Y] + \Pr[X\leq Y] = 1 + {\sf CP}(X)\leq 1 + \frac{1}{2^k}\;,
	\end{displaymath}
	where the last inequality uses the standard upper bound on the collision
	probability ${\sf CP}(X)$ with respect to the min-entropy of $X$.
\end{proof}

\end{document}